\documentclass[conference]{IEEEtran}
\IEEEoverridecommandlockouts
\usepackage{cite}
\usepackage{amsmath,amssymb,amsfonts}
\usepackage{graphicx}
\usepackage{textcomp}
\usepackage{xcolor}
\usepackage{algorithm}
\usepackage{algorithmic}
\usepackage{tikz}
\usetikzlibrary{shapes.geometric, arrows}
\usepackage{enumitem}
\usepackage{inconsolata}
\usepackage{subcaption}
\usepackage{threeparttable}
\usepackage{multirow}
\usepackage{comment}
\usepackage{csquotes}
\usepackage{amsthm}
\usepackage{setspace}
\newtheoremstyle{italiclabel} 
  {\topsep}   
  {\topsep}   
  {\normalfont}  
  {}          
  {\itshape}  
  {:}         
  {.5em}      
  {}          

\theoremstyle{italiclabel}
\newtheorem{theorem}{Theorem}

\newtheorem{remark}{Remark}
\newtheorem{proposition}{Proposition}

\usepackage[font=normal]{caption} 
\usepackage{fancyhdr}
\usepackage{caption}
\def\BibTeX{{\rm B\kern-.05em{\sc i\kern-.025em b}\kern-.08em
    T\kern-.1667em\lower.7ex\hbox{E}\kern-.125emX}}
\begin{document}

\title{A Machine Learning-Based Reference Governor for Nonlinear Systems With Application to Automotive Fuel Cells \\
}
\author{
    Mostafaali Ayubirad\textsuperscript{*}, \textit{Student Member, IEEE}, and Hamid R. Ossareh, \textit{Senior Member, IEEE} \thanks{\textsuperscript{*} Corresponding author: Mostafaali Ayubirad}\thanks{The authors are with the Department of Electrical and Biomedical Engineering, University of Vermont, Burlington, VT 05405 USA (e-mail: {\tt mostafa-ali.ayubirad@uvm.edu;
hamid.ossareh@uvm.edu)}.} \thanks{© 2025 IEEE. Personal use of this material is permitted. Permission from IEEE must be obtained for all other uses.}
}

\maketitle

\begin{abstract}
The prediction-based nonlinear reference governor (PRG) is an add-on algorithm to enforce constraints on pre-stabilized nonlinear systems by modifying, whenever necessary, the reference signal. The implementation of PRG carries a heavy computational burden, as it may require multiple numerical simulations of the plant model at each sample time. To this end, this paper proposes an alternative approach based on machine learning, where we first use a regression neural network (NN) to approximate the input-output map of the PRG from a set of training data. 
During the real-time operation, at each sample time, we use the trained NN to compute a nominal reference command, which may not be constraint admissible due to training errors and limited data. We adopt a novel sensitivity-based approach to minimally adjust the nominal reference while ensuring constraint enforcement. We thus refer to the resulting control strategy as the  modified neural network reference governor (MNN-RG), which is computationally more efficient than the PRG.  
The computational and theoretical properties of MNN-RG are presented. Finally, the effectiveness and limitations of the proposed method are studied by applying it as a load governor for constraint management in automotive fuel cell systems through simulation-based case studies.
\end{abstract}


\section{Introduction}
Constraint enforcement in control systems is a critical aspect that has garnered significant attention in recent years. The ability to ensure that system variables remain within prescribed bounds is essential for guaranteeing safe and optimal operation in various industrial processes, robotics, and autonomous systems \cite{del2010automotive, sun2018robust, garone2017reference}. One approach to enforcing constraints in control systems is through Model Predictive Control (MPC) \cite{rawlings2017model}. While MPC offers a comprehensive framework for managing constraints, it requires solving a real-time optimization problem, making it computationally demanding, particularly for systems with fast dynamics and/or nonlinearities. 
As an alternative approach, Lyapunov and barrier functions can be combined and shaped to enforce various types of constraints (state, input, output, etc.) on the system. However, synthesizing appropriate Lyapunov/barrier functions can be challenging, especially for complex systems with multiple constraints, and the resulting control actions may be overly conservative \cite{anand2021safe}.

To address the limitations of MPC and barrier function methods, the Reference Governor (RG) \cite{gilbert1995discrete} strategy was proposed. As its name suggests, the RG enforces the constraints by modifying the reference command of a pre-stabilized closed-loop system. A detailed survey on techniques based on RGs can be found in \cite{garone2017reference,kolmanovsky2014reference}. Much of the RG research has centered around linear systems, where it has been shown that RG is computationally less demanding as compared to MPC and is simpler to design than barrier methods. 

To extend the applicability of the RG approach to nonlinear systems, several modifications and extensions have been proposed in the literature. These approaches generally fall into two categories. The solutions in the first category use linear RGs in a nonlinear setting by exploiting feedback linearization \cite{kalabic2015reference}, embedding a nonlinear model into a family of switched/hybrid linear models \cite{franze2014reconfigurable}, or designing an RG based on an approximated linear model where a disturbance bias is used to compensate for linearization errors \cite{vahidi2006constraint}. While effective in certain applications, these approaches either lack guarantees of strong properties like constraint enforcement (see, e.g.,\cite{vahidi2006constraint}) or face challenges in embedding complex nonlinear systems into the Linear Parameter-Varying (LPV) framework (see, e.g.,\cite{franze2014reconfigurable}). 

The solutions in the second category involve nonlinear RGs that are designed directly based on the nonlinear system models as opposed to  linear ones 
\cite{angeli1999command,bemporad1998reference,gilbert2002nonlinear,gilbert1999set,miller2000control}. In \cite{miller2000control}, the authors leverage the invariance of the level sets of Lyapunov functions to construct an RG scheme \cite{miller2000control}. This approach is computationally efficient but tends to be conservative, potentially resulting in slow response. In \cite{bemporad1998reference}, the author introduces a bisectional search combined with predictive online numerical simulations of the system dynamics to find an optimal constraint-admissible reference. This approach, shown in Fig.~\ref{fig:1}(a), is refereed to as the prediction-based nonlinear RG (PRG). One of the challenges with PRG is that it is computationally intensive due to the repeated online simulations mentioned above. 

As a way to improve the computational efficiency of PRG, in our earlier work in \cite{lim2023reference}, we trained a regression neural network (NN)  offline to approximate the PRG's input-output map.
The trained NN was then used in real-time to compute the reference command as shown in Fig.~\ref{fig:1}(b). We refer to this approach as the Neural Network Reference Governor (NN-RG) and leverage it in our work. It was shown that with a sufficient dataset, the NN-RG can achieve satisfactory performance and is computationally more efficient than the PRG; however, unlike the PRG, it does not provide theoretical guarantees for constraint enforcement 
due to the inherent approximation errors of the NN.

To fill these gaps, we propose a modification of the NN-RG, shown in Fig.~\ref{fig:1}(c), which we refer to as the modified NN-RG (MNN-RG). 
Specifically, at each sample time, we compute the output of the NN. We then predict the output trajectory of the nonlinear system driven by this nominal command. In addition, we compute a “sensitivity” function, which describes the sensitivity of the output to changes in the reference. With these two, we use a first-order Taylor series with a bounded remainder term, and formulate a novel RG, which is robust to the resulting linearization errors. 
In other words, in the proposed approach, the output of the NN is modified to eliminate any potential constraint violations caused by the NN training errors and/or data collection issues. In this paper, we first introduce the MNN-RG, and then show its desirable properties of constraint satisfaction and closed-loop stability. We then present practical and computationally efficient methods for bounding the Taylor remainder term discussed above.

\begin{figure}
    \centering
   \begin{subfigure}{0.4\textwidth}
        \centering
        \includegraphics[width=\textwidth, height=0.23\textwidth, keepaspectratio=false]{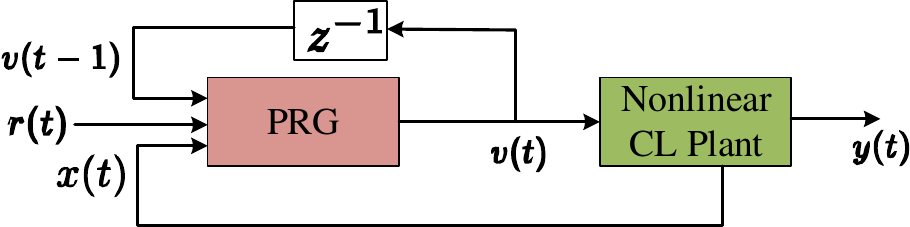}
        \caption{PRG}
    \end{subfigure}
    \begin{subfigure}{0.45\textwidth}
        \centering
        \vspace{0.3em} 
        \includegraphics[width=\textwidth, height=0.23\textwidth, keepaspectratio=false]{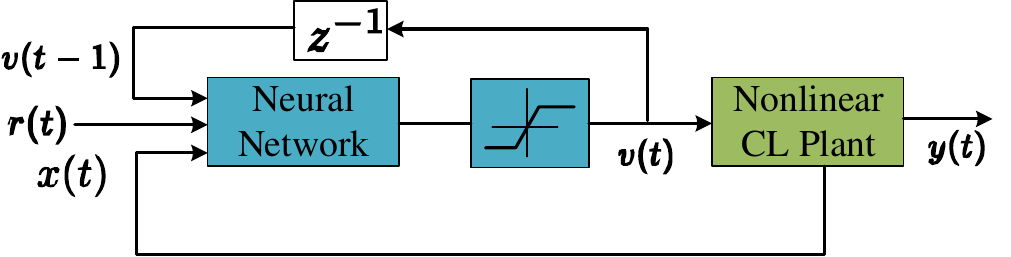}
        \caption{NN-RG}
    \end{subfigure} 
    \begin{subfigure}{0.49\textwidth}
        \centering
        \vspace{0.3em} 
        \includegraphics[width=\textwidth, height=0.25\textwidth, keepaspectratio=false]{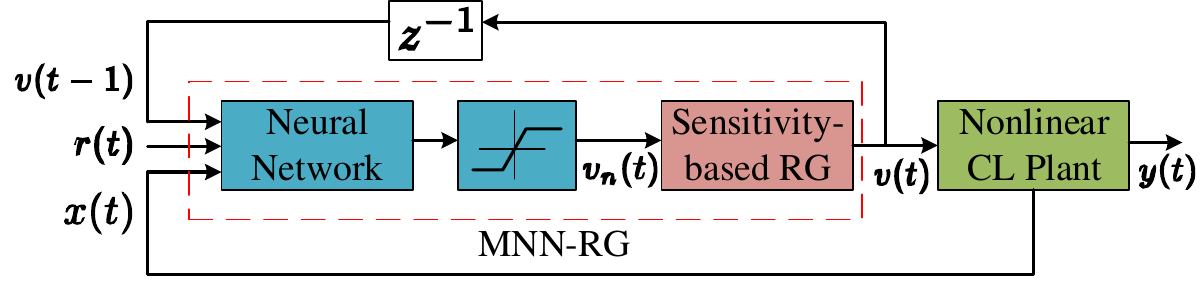}
        \caption{MNN-RG}
    \end{subfigure} 
    \caption{Reference governor controller schemes. The signals are as follows: $y(t)$ is the constrained output, $r(t)$ is the desired reference, $v(t)$ is the modified reference command, $v_{n}(t)$ is the nominal reference, and $x(t)$ is the system state.}
    \label{fig:1}
\end{figure}

Both the NN-RG and MNN-RG require an explicit model, i.e., a mathematical representation of the plant dynamics, as well as knowledge of the constraints, since their training relies on data from PRG, which necessitates this information. There are other learning-based RG approaches in the literature, which are referred to as Learning Reference Governors (LRGs), which learn from data, such as a black-box model or actual system experimentation, eliminating the need for explicit model knowledge
\cite{ikeya2022learning,ong2004stability, liu2021safe,li2019dynamics,lanchares2019reference}. Under some assumptions, these LRGs can  theoretically guarantee constraint satisfaction \cite{liu2021safe}, learn system constraints \cite{li2019dynamics}, and distinguish safe pairs of states and reference commands from the unsafe ones \cite{lanchares2019reference}. However, LRGs require an online training phase during system operation, which could be at times a lengthy process and may lead to constraint violations during training. The NN-RG and MNN-RG do not suffer from these limitations.

In the latter half of this paper, the proposed approach is applied as a load governor to an automotive fuel cell (FC) to prevent oxygen starvation and avoid compressor surge and choke regions. This is not only an illustration of the proposed approach but also serves as a second contribution of this paper. 

Proton exchange membrane fuel cell (PEMFC) systems, which produce electricity through electrochemical reactions, have garnered attention for automotive applications due to their high efficiency and zero emissions. The problem of FC control has been extensively studied, interested readers can refer to \cite{larminie2003fuel,pukrushpan2002modeling,goshtasbi2020degradation,ayubirad2024model}. Due to limited onboard computational resources in automotive vehicles, RGs have found numerous applications for constraint enforcement of automotive FC systems (see, e.g., \cite{ayubirad2023simultaneous, bacher2022efficiency, bacher2023hierarchical,sun2005load,vahidi2006constraint}). Specifically, in \cite{sun2005load}, the PRG was utilized to design a robust RG for the FC's oxygen starvation constraint. In \cite{vahidi2006constraint}, a fast linear RG with low computational requirements was employed to enforce the compressor surge, choke, and oxygen starvation constraints. Despite its efficiency, the linear RG failed to completely eliminate constraint violations and 
small oscillations around the constraint boundary
in the response of the nonlinear FC model. In our paper, we demonstrate the use of the MNN-RG as a load governor in FCs, highlighting its low computational burden compared to the PRG, while ensuring no constraint violations. We also study the tuning and various nuances of applying MNN-RG to the FC constraint management problem.

In summary, this paper contributes to both the control theory and automotive FC literature as follows:

\textbf{The original contributions to control literature:}
The MNN-RG is proposed with the potential to require  fewer training samples  for offline training than the NN-RG, guarantee constraint satisfaction by modifying the NN output through a sensitivity function, and reduce online execution time  compared to the PRG. Theoretical guarantees 
including constraint satisfaction, stability, and convergence are provided. Finally, methods for tuning the MNN-RG in practice are presented.

 \textbf{The original contributions to FC literature:}
A computationally efficient load governor based on MNN-RG is developed to enforce constraints on compressor surge, choke, and oxygen starvation for FC systems. Additionally, methods for generating the necessary dataset to train the NN used in the MNN-RG-based load governor are provided.

The paper is organized as follows. Section~\ref{Preliminaries} contains a review of nonlinear reference governors, including the NN-RG and PRG. The novel MNN-RG is proposed in Section~\ref{Modified neural network reference governor (MN-RG)}. Section~\ref{Implementation of the modified NNRG on the FC air-path system} presents the MNN-RG application to 
FCs. Finally, conclusions are drawn in Section~\ref{conclusions}.

\section{Preliminaries}
\label{Preliminaries}
In this section, we first present the problem setup, followed by a review of the PRG and NN-RG, as these are leveraged in our proposed approach.

\subsection{Problem setup}
\label{Problem setup}
Consider Fig.~\ref{fig:1}, in which the ``Nonlinear CL Plant'' represents a general closed-loop nonlinear system described by
\begin{flalign}
        x(t+1)=f(x(t),v(t))
\label{eqn:1}
\end{flalign}
and subjected to constraints,
\begin{flalign}
y(t)=h(x(t),v(t))\le 0
\label{eqn:2}
\end{flalign}
where $x(t)\in {\mathbb{R}}^{n_x}$ is the state vector, $v(t)\in \mathbb{R}$ is the modified reference command, and the constraint output $y(t)$ is a scalar (this assumption will be relaxed later in the paper when MNN-RG for multi-output systems is presented).
We make the following assumptions about system~\eqref{eqn:1}--\eqref{eqn:2}:
\begin{enumerate}[label=(\alph*)]

\item The input satisfies \( v(t) \in \mathcal{V} \), where \( \mathcal{V} \) is compact. Since the input is bounded, we assume that the state is also bounded, i.e., \( x(t) \in \mathcal{X} \), where \( \mathcal{X} \) is compact.

\item The function \( f : \mathbb{R}^{n_x} \times \mathbb{R} \to \mathbb{R}^{n_x} \), and the function \( h : \mathbb{R}^{n_x} \times \mathbb{R} \to \mathbb{R} \) are twice continuously differentiable on \( \mathcal{X} \times \mathcal{V} \).
\item For a constant input \( v \), the equilibrium manifolds \( {\bar{x}_v = \lim_{t \rightarrow \infty} x(t) }\) and \( \bar{y}_v = \lim_{t \rightarrow \infty} y(t) \), which depend only on \( v \), exist and are uniquely-defined.
\item For all $\lambda > 0$, there exists $\alpha(\lambda) > 0$, such that 
${\| x(0) - \bar{x}_v \| \leq \alpha(\lambda)}$ implies $\| x(t) - \bar{x}_v \| \leq \lambda$ 
for all $t \geq 0$ and all $v \in \mathcal{V}$.
\end{enumerate}
Assumption (a) is nonrestrictive since in practice the references and state variables are bounded. Assumption (b) ensures the existence of sensitivity functions and a well-defined, bounded remainder term in the first-order Taylor series, as will be discussed later in Section~\ref{Modified neural network reference governor (MN-RG)}. Assumption (c) establishes $\bar{x}_v$ and $\bar{y}_v$ as asymptotic solutions for the system \eqref{eqn:1}--\eqref{eqn:2} under constant input $v$. Lastly, assumption (d) guarantees uniform stability of the system across all constant inputs $v$. Assumptions (a)--(d) are required for the proof of finite constraint horizon property in \cite{bemporad1998reference}, which we leverage later in our work.

The Maximal Output Admissible Set (MOAS), denoted by \( O_{\infty} \), is the set of all initial state $x_0$ and constant inputs $v_0$, which satisfies the constraints for all future times, i.e,
\begin{align}
    O_{\infty} = \left\{ (x_0, v_0) : x(0) = x_0, \ v(j) = v_0, \right. & \nonumber \\
    \left. \Rightarrow y(j|x_0, v_0) \leq 0, \ j \in \mathbb{Z}_+ \right\}
    \label{eqn:3}
\end{align}
where ${y}(j|x_0,v_0)$ is the  output trajectory given $x_0$ and $v_0$, and $\mathbb{Z}_+$ denotes the set of all nonnegative integers. The set \( O_{\infty} \) is compact. Specifically, it is closed due to the continuity of \( f \) and \( h \) (Assumption (b)), and bounded in \( \mathcal{X} \times \mathcal{V} \) due to the compactness of \( \mathcal{X} \) and \( \mathcal{V} \) (Assumption (a)). According to the finite constraint horizon property in \cite{bemporad1998reference}, given the assumptions (a)–(d) and the compactness of \( O_{\infty} \), with a tightened constraint on the output at steady-state, i.e., \( \bar{y}_v \leq -\varepsilon \), where \( \varepsilon > 0 \), there exists a \( j^* \) such that \( y(j) \leq 0 \) for \( j = 0 \) to \( j^* \) implies \( y(j) \leq 0 \) for all \( j > j^* \).
This implies that a finitely determined inner approximation of \( O_{\infty} \) can be constructed as
\begin{flalign}
    \begin{split}
        \mathrm{\Omega }=\left\{\left(x_0,v_0\right)\ :\ {x}(0)=x_0,\ \ v(j)=v_0,\ \right. &\\
        \Rightarrow  \overline{y}_v \le -\varepsilon \ ,\ {y}(j|x_0,v_0) \le 0\ ,\ j=0,\dots ,j^*\left.\right\}
    \end{split}
    \label{eqn:4}
\end{flalign}
where $\Omega$ can be made arbitrary close to \( O_{\infty} \) by decreasing $\varepsilon $ at the expense of increasing the complexity of $\Omega$.

The implicit solution for $j^*\ $can be obtained using the level set of a Lyapunov function \cite{garone2017reference}. However, in this paper we will use the approach in \cite{sun2005load}, where $j^*$ is chosen 2 to 5 times larger than a system's time constant. Therefore, the $j^*$ we are considering here is not necessarily the smallest such $j^*$.

\subsection{Review of prediction-based nonlinear reference governor (PRG)}
\label{Review of prediction-based nonlinear reference governor}
Consider the block diagram in Fig. \ref{fig:1}(a), where the desired setpoint $r(t)$ satisfies $r(t)\in \mathcal{V}$, an assumption that we will continue to make for the rest of this paper.
To prevent constraint violation, at each timestep \( t \), the PRG selects the optimal control input \( v(t) \) such that the pair \((x(t), v(t)) \in \Omega\) and \( v(t) \) is as close as possible to \( r(t) \). More formally, the PRG calculates the $v(t)$ by implementing the following dynamic equation: 
\begin{flalign}
v(t)=v(t-1)+\kappa (t)(r(t)-v(t-1))
\label{eqn:5}
\end{flalign}
where $\kappa(t)\in[0,1]$ is obtained by solving the following optimization problem:
\begin{flalign} 
    \begin{aligned}
        \kappa (t)={\mathop{\max }\limits_{\kappa \in [0,1]}} {\rm \; }\kappa\\
        \text{s.t.} \quad &v=v(t-1)+\kappa \left(r(t)-v(t-1)\right)\\ 
        &\hat{x}(j+1)=f(\hat{x}(j),v) \\
        &h(\hat{x}(j),v)\le 0,{\rm \; }j=0,...,j^{*}\\
        &h(\bar{x}_{v} ,v)\le -\varepsilon
    \end{aligned}
\label{eqn:6}
\end{flalign}
where \(\hat{x}(j) = \hat{x}(t+j|t)\) represents the predicted state trajectory propagating from \(\hat{x}(0) = x(t)\),
where $x(t)$ is the full-state measurement, which we assume to be available.
For simplicity, we will use the notation \(\hat{x}(j)\) throughout the rest of the paper to refer to the \(j\)-step ahead prediction, without explicitly referencing \(t\). Note that the predicted states $\hat{x}(j)$ within the RG differ from the actual closed-loop system states in \eqref{eqn:1}, as the RG computes its predictions at each timestep assuming a constant input $v$ over the prediction horizon, while in the actual system, $v$ may change at each timestep.

The optimization in \eqref{eqn:6} is generally nonconvex and can be solved by bisectional or grid search methods.  For the sake of clarity and completeness, the search method is summarized in Algorithm 1, where the PRG is implemented through a bisection search to find a near-optimal $\kappa$. In Algorithm 1, $L$ controls the number of iterations the algorithm performs to find a near-optimal value for $\kappa$. The bisectional search may require multiple simulations of the nonlinear model (step 6 of Algorithm 1) at each timestep, which is why the computational cost of the PRG algorithm is high.

The high computational burden of the PRG may prohibit its the real-time implementation. It has been suggested that replacing the PRG with a nonlinear function approximator, such as a Neural Network (NN), might make it suitable for online implementation (as shown in Fig.~\ref{fig:1}(b)) \cite{sun2005load}. However, the limited availability of data for training the NN, along with the inevitable approximation errors of the NN, constitutes a problem that we will address in this paper.
\begin{algorithm}[t]
\caption{The PRG algorithm to compute \( \kappa \)}
\label{alg:kappa_opt}
\begin{algorithmic}[1]
\STATE \textbf{Input:} $x(t)$, $v(t-1)$, $r(t)$, $L$, $j^*$, $\varepsilon$
\STATE \textbf{Output:} $\kappa_{\text{opt}}$
\STATE Initialize $\underline{\kappa} \gets 0$, $\overline{\kappa} \gets 1$, $\kappa \gets 1$, $\kappa_{\text{opt}} \gets 0$
\FOR{$i = 1$ to $L$}
    \STATE $v \gets v(t-1) + \kappa(r(t) - v(t-1))$
    \STATE Compute $\hat{y}(j)$, $j = 0, \ldots, j^*$, and $\bar{y}_v$ for $(x(t), v)$
    \IF{$\max(\hat{y}(j)) \leq 0$ and $\bar{y}_v \leq -\varepsilon$}
        \STATE $\kappa_{\text{opt}} \gets \kappa$
        \IF{$\kappa = 1$}
            \STATE \textbf{break}
        \ELSE
            \STATE $\underline{\kappa} \gets \kappa$, $\kappa \gets \left(\underline{\kappa} + \overline{\kappa}\right) / 2$
        \ENDIF
    \ELSE
        \STATE $\overline{\kappa} \gets \kappa$, $\kappa \gets \left(\underline{\kappa} + \overline{\kappa}\right) / 2$
    \ENDIF
\ENDFOR
\end{algorithmic}\label{alg:PRG}
\end{algorithm}

\subsection{Review of Neural Network Reference Governor (NN-RG)}
\label{Neural network reference governor (NN-RG)}
NNs are universal function approximators that are commonly used, which motivates their use in approximating the PRG function. The existing works (e.g., \cite{lim2023reference}) implicitly assume that the PRG can be approximated by an NN with reasonable accuracy.
The main idea behind the NN-RG is to replace the PRG described previously with a well-trained NN, as shown in Fig.~\ref{fig:1}(b), thus bypassing the need to simulate the nonlinear model multiple times in real-time, leading to a significant reduction in computation time. The NN is trained using numerous observations of input data $[x(t),v(t-1),r(t)]^\top$ and the corresponding output $v(t)$. {The training dataset is generated by executing the PRG on the nonlinear system under various scenarios, thus helping to create an NN that generalizes well. Since collecting the training dataset requires implementing the PRG, we assume one of two scenarios in this work: (1) a PRG with satisfactory performance can be run on real hardware, which would require a fast processor and involve a non-trivial implementation effort; or (2) a simulator is available, in which case execution time is less of a concern. Once trained, the NN replaces the PRG during real-time operation. It is evaluated at each timestep to generate the reference signal for the closed-loop system, enabling deployment on an embedded processor that is less capable  and more cost-effective than what would be required to run the PRG directly.}

In this paper, we slightly enhance the NN-RG described in \cite{lim2023reference} by making sure that $v(t)$ remains  between \(v(t-1)\) and  \(r(t)\), similar to PRG. Specifically, we introduce a scaling factor \(\kappa_{NN}(t)\), computed as:
\begin{flalign}
    \kappa_{NN} (t)=\max \left(\min \left(\frac{v(t)-v(t-1)}{r(t)-v (t-1)} ,1\right),0\right)
\label{eqn:7}
\end{flalign}
The scaling factor is then used to adjust the NN output as \(v(t) = v(t - 1) + \kappa_{NN}(t)(r(t) - v(t - 1))\) before it is applied to the system.  The above logic is represented by the dynamic saturation block in Fig. \ref{fig:1}b. Note that ideally, \(\kappa_{NN}(t)\) would equal \(\kappa(t)\) from the PRG in Eq. \eqref{eqn:6}; however, due to inherent approximation errors of NNs, this is not always the case.

While the concept of replacing the PRG with an NN is intriguing, it is crucial to note that this approach may result in constraint violations due to approximation errors in the NN. The MNN-RG approach proposed here overcomes this limitation.

\section{Modified neural network reference governor (MNN-RG)}
\label{Modified neural network reference governor (MN-RG)}
In this section, we provide the MNN-RG formulation and its theoretical properties (Section~\ref{MNN-RG design}), discuss a tuning process for the MNN-RG that is practical (Section~\ref{MNN-RG tuning}), and outline a complete design process (Section~\ref{MNN-RG design process}).

\subsection{MNN-RG formulation and theoretical properties}
\label{MNN-RG design}
As mentioned previously, the inherent approximation error in NN impacts the tracking and constraint satisfaction properties of NN-RG. This motivates us to pursue modifications of the NN output through MNN-RG. Interestingly, while the existing works on NN-RG make the implicit assumption that the PRG can be approximated by an NN with high accuracy, this assumption is no longer needed for MNN-RG. As we show later in Section~\ref{Implementation of the MN-RG}, even a with poorly generalized NN, the MNN-RG can enforce constraints and result in good performance.

As shown in Fig.~\ref{fig:1}(c), the high-level idea of the MNN-RG involves first computing \( v_n \), which is the output of a trained NN, saturated between $v(t-1)$ and $r(t)$, similar to NN-RG.  The MNN-RG then solves \eqref{eqn:1}--\eqref{eqn:2} using \( v_n \) as a constant input to obtain the ``nominal predicted output''. A sensitivity function is then employed to compute a truncated Taylor series expansion of the output prediction around the nominal predicted output trajectory, using the remainder term to bound the truncation error. Finally, we run an algorithm similar to a linear RG, but robustified to account for this truncation error, ensuring constraint enforcement.
Below, we elaborate on these ideas and discuss how the remainder term can be bounded in practice.

To begin, note that, at each time instant $t$, the prediction of the output starting from the initial condition $x(t)$ is obtained by solving the following difference equation:
\begin{flalign}
    \begin{split}
        \hat{x}_{n}(j+1) &= f(\hat{x}_{n}(j), v_{n}) \\ 
        \hat{y}_{n}(j) &= h(\hat{x}_{n}(j), v_{n}) \\
        \hat{x}_{n}(0) &= x(t)
    \end{split}
\label{eqn:8}
\end{flalign}
where $j = 0, \ldots, j^*$, and the hat notation denotes predicted values. As before, the input is assumed to be constant over the prediction horizon, i.e.,  $v(j)=v_n$, $j=0,\ldots,j^*$. Since $\hat{y}(j)$ changes depending on the value of constant $v$, we expand $\hat{y}(j)$ in a Taylor series around the nominal prediction trajectory $\hat{y}_{n}(j)$:
\begin{flalign}
    \begin{split}
        \hat{y}(j)=\hat{y}_{n} (j)+S_{y} (j)(v-v_{n} )+R_v(j),\ j=0,...,j^{*}
    \end{split}
\label{eqn:9}
\end{flalign}
where ${S_{y}(j)}={\frac{dy}{dv}}|_{\substack{x=\hat{x}_n(j)\\v=v_n} }$  is the sensitivity function of ${y}$ with respect to $v$, evaluated at $v=v_n$ along the nominal predicted state trajectories $\hat{x}_n(j)$. The term $R_v(j)$ is the remainder term in the Taylor series , which can be written as:
\begin{flalign}
    \begin{split}
       R_v(j) &= S_y^{(2)}\frac{(v - v_n)^2}{2!} + S_y^{(3)} \frac{(v - v_n)^3}{3!} + \dots
    \end{split}
\label{eqn:R_v}
\end{flalign}
where \( S_y^{(i)}(j) = \frac{d^i y}{dv^i} \) evaluated at \( (\hat{x}_n(j), v_n) \), for all \( i \in \mathbb{Z}_{>0} \), with \( \mathbb{Z}_{>0} \) denoting the set of all positive integers. For a nominal input \( v_n \) and nominal state trajectory \( \hat{x}_n(j) \), the remainder term \( R_v(j) \) explicitly depends on \( v \), while its coefficients are evaluated at \( (\hat{x}_n(j), v_n) \). This implies that \( R_v(j) \) is indirectly influenced by the nominal state trajectory and input, with the effect of state dynamics on \( R_v(j) \) captured implicitly rather than explicitly through \( S_y^{(i)}(j) \).
According to the law of total derivatives, the $S_{y}(j)$ in (\ref{eqn:9}) is calculated as:
\begin{flalign}
S_{y} (j) & = \frac{\partial h(x,v)}{\partial x} \Big|_{\substack{x=\hat{x}_n(j)\\v=v_n}} S_{x} (j) + \frac{\partial h(x,v)}{\partial v}  \Big|_{\substack{x=\hat{x}_n(j)\\v=v_n}} 
\label{eqn:10}
\end{flalign}
where $S_x(j)=\frac{ dx}{ dv}|_{\substack{x=\hat{x}_n(j)\\v=v_n} }$ \ is the sensitivity of $x$ with respect to $v$, which can be computed as a solution to the following difference equation [\citenum{khalil2001nonlinear}, page 99]:
\begin{flalign}
    \begin{split}
        S_{x} (j+1)&= \frac{\partial f(x,v)}{\partial x}  \Big|_{\substack{x=\hat{x}_n(j)\\v=v_n}} S_{x} (j) + \frac{\partial f(x,v)}{\partial v} \Big|_{\substack{x=\hat{x}_n(j)\\v=v_n}} \\
        S_{x}(0) & = 0
    \end{split}
    \label{eqn:11}
\end{flalign}
Following Assumptions (a)–(b) in Section \ref{Problem setup}, where \( f \) and \( h \) are assumed to be twice continuously differentiable on the compact set \( \mathcal{X} \times \mathcal{V} \), the function \( \frac{d^2 y}{dv^2} \) satisfies the conditions required by the Extreme Value Theorem. Therefore, there exists an upper bound \( M \in \mathbb{R}_{\geq 0}\) such that ${\sup_{\substack{(x,v) \in (\mathcal{X} \times \mathcal{V})}} \left| \frac{d^2 y}{dv^2} \right| \leq M}
$. Following this, by Taylor’s theorem, the remainder term \( R_v(j) \) in a trajectory-based Taylor expansion of the solution is bounded by \( R_v(j) \leq \frac{M}{2} (v - v_n)^2 \).
For now, we assume that $M$ is known. The  computation of $M$ in the context of MNN-RG will be discussed later. With this upper bound on \( R_v(j) \), the output prediction $\hat{y}(j)$ can be upper bounded as follows:
\begin{flalign}
    \begin{split}
        \hat{y}(j) \leq 
        \hat{y}_{n} (j)+S_{y} (j) (v-v_{n} )+\frac{M}{2}{\left(v-v_n\right)}^2
    \end{split}
\label{eqn:12}
\end{flalign}
for $j=0,...,j^{*}$. Thus, the constraint \(\hat{y}(j) \leq 0\) over the prediction horizon can be replaced by the following tightened condition:
\begin{flalign}
    \begin{split}
        \hat{y}_{n} (j)+S_{y} (j) (v-v_{n} )+\frac{M}{2}{\left(v-v_n\right)}^2\le 0
    \end{split}
\label{eqn:13}
\end{flalign}
Because of \eqref{eqn:12}, if $v$ satisfies \eqref{eqn:13}, then $\hat{y}(j)\leq 0$  as desired, i.e., the constraints are satisfied during the prediction horizon and, as discussed in Section~\ref{Preliminaries}, over all future times.
Thus, the ``Sensitivity-based RG'' in Fig.~\ref{fig:1}(c) is formulated as the following optimization problem:
\begin{flalign} 
    \begin{aligned}
        \kappa (t)={\mathop{\max }\limits_{\kappa \in [0,1]}} {\rm \; }\kappa\\
        \text{s.t.} \quad &v=v(t-1)+\kappa (r(t)-v(t-1))\\ 
        & \hat{y}_{n} (j)+S_{y} (j)(v-v_n )\\
        &+\frac{M}{2} (v-v_{n} )^{2} \le 0, \ j=0,...,j^{*}\\ 
        &\bar{y}_{v} \le -\varepsilon
    \end{aligned}
\label{eqn:14}
\end{flalign}
Upon finding the optimal $\kappa$, $v(t)$ in Fig.~\ref{fig:1}(c) is then found using the standard RG law in \eqref{eqn:5}, which is guaranteed to enforce the constraints, as discussed in the following theorem:
\begin{theorem}
Let $v(-1)$ be an initial command and $x(0)$ an initial state at time $t=0$ such that the pair $(x(0),v(-1))\in \Omega$, where $\Omega$ is as defined in \eqref{eqn:4}. Let the assumptions in Section~\ref{Problem setup} hold and assume that the dynamics in \eqref{eqn:1} are Input-to-State Stable (ISS). Then, the MNN-RG guarantees constraint enforcement on the nonlinear system for all subsequent times $t\geq0$, ensures the ISS stability for the entire system composed of the MNN-RG and the closed-loop system, and for a constant $r$, $v(t)$ converges to a constant.
\end{theorem}
\begin{proof}\
See Appendix A.
\end{proof}
We now discuss how explicit solutions to \eqref{eqn:14} can be computed. Notice that the constraints in the optimization problem in \eqref{eqn:14} can be quadratic in $\kappa$, similar to \cite{osorio2022novel}. Specifically, the constraints are of the form:
\begin{flalign}
a_2(j)\kappa^2 + a_1(j)\kappa + a_0(j) \leq 0, \ j=0,...,j^{*}
\label{eqn:15}
\end{flalign}
where
\begin{flalign*}
\begin{split}
    a_2(j) &= \frac{M}{2} (r(t) - v(t-1))^{2} \\
    a_1(j) &= S_{y}(j) (r(t) - v(t-1)) \\
    &\quad + M (r(t) - v(t-1)) (v(t-1) - v_{n}) \\
    a_0(j) &= \hat{y}_{n}(j) + S_{y}(j) (v(t-1) - v_{n})
\end{split} 
\end{flalign*}
  
Note that $a_2(j)$ in (16) is in fact constant; that is $a_2(j)=a_2$. The solution to \eqref{eqn:15} is categorized into two cases based on whether \( a_2 > 0 \) or \( a_2 = 0 \). For \( a_2 > 0 \), let the roots of the \(j\)-th quadratic inequality in \eqref{eqn:15} be denoted by \(\kappa_{L,j}\) and \(\kappa_{U,j}\). If the roots are complex, then \eqref{eqn:15} is infeasible, and we set \(\kappa\) to zero, thereby guaranteeing constraint satisfaction. If the roots are real, the range of acceptable \(\kappa\) will be between the two roots, i.e., \(\kappa \in [\kappa_{L,j}, \kappa_{U,j}]\). Thus, the explicit solution for \(\kappa\) in \eqref{eqn:15} is: 
\begin{flalign}
\kappa =
\begin{cases}
    \kappa_{U} & \text{if all } \kappa_{U,j}, \kappa_{L,j} \text{ are real} \text{ and } \kappa_{L} \leq \kappa_{U} \\
    0 & \text{otherwise}.
\end{cases}
\label{eqn:kappe_quadratic}
\end{flalign}
where 
\[
    \kappa_{U} = \min \left\{ \min_{j} \left(\kappa_{U,j}\right), 1 \right\}, \quad \kappa_{L} = \max \left\{ \max_{j} \left(\kappa_{L,j}\right), 0 \right\}.
\]    
For \( a_2 = 0 \), the explicit solution of \eqref{eqn:15} is:
\begin{flalign}
\kappa =
\begin{cases}
    \kappa_{U} & \text{if } \kappa_{L} \leq \kappa_{U} \text{ and } a_0(j) \leq 0 \text{ for all } j \text{ such that }\\ &a_1(j) = 0, \\
    0 & \text{otherwise}
\end{cases}
\label{eqn:kappe_linear}
\end{flalign}
where 
\begin{align*}
\kappa_U &= \min \left\{ \min_{j \,:\, a_1(j) > 0} \left( -\frac{a_0(j)}{a_1(j)} \right), 1 \right\}, \\  
\kappa_L &= \max \left\{ \max_{j \,:\, a_1(j) < 0} \left( -\frac{a_0(j)}{a_1(j)} \right), 0 \right\}.
\end{align*}
\begin{remark}
\label{rm:twowaysolution}
An alternative approach to computing \( \kappa \) in \eqref{eqn:15} is to perform a bisectional search, similar to Algorithm 1, to determine the largest \( \kappa \) satisfying the inequalities in \eqref{eqn:15}.
The formula in \eqref{eqn:kappe_quadratic}-\eqref{eqn:kappe_linear} provide an exact solution to \eqref{eqn:15}; however, for quadratic inequalities, \eqref{eqn:kappe_quadratic} is only recommended when the number of inequalities is small. Otherwise, the bisectional method is preferred, as it provides an approximate solution to \eqref{eqn:15} but is easy to implement, generally robust, and less prone to numerical instability, making it preferable when dealing with a large number of inequalities.
\end{remark}
\begin{remark}
The MNN-RG optimization problem in \eqref{eqn:14} is not necessarily recursively feasible. However, this does not pose a problem, as for a constant input computed from \eqref{eqn:14} at each timestep, the constraint \(\hat{y}(j) \leq 0\) will be satisfied for all time. Proving recursive feasibility for \eqref{eqn:14} is challenging because \(S_y(j)\), which depends on \(v_n\) of the NN, changes at every timestep. However, our extensive numerical studies suggest that MNN-RG does not suffer from deadlocks due to lack of recursive feasibility.
\end{remark}
\begin{remark} 
We remark on the connection between \eqref{eqn:14} and the standard linear RG and show the latter is equivalent to \eqref{eqn:14} with $M=0$.
Suppose system \eqref{eqn:1} is a pre-compensated discrete-time linear system as follows:
\begin{flalign}
    \begin{split}
        x(t+1)&=A x(t)+Bv(t)\\ 
        y(t)&=Cx(t)+Dv(t)
    \end{split}
\label{eqn:17}
\end{flalign}
By using the linear model in \eqref{eqn:17} and assuming a constant command $v$, the constraint $\hat{y}(j) \le 0$ in the standard linear RG is expressed in the form $H_x(j)x\left(t\right)+H_v(j)v \le 0$, where $x(t)$ is the currently available state, $H_x(j)=CA^{j}$ and$\ H_v(j)=C(I-A^{-1})(I-A^{j})B+D$. This constraint is exactly the MNN-RG constraint in \eqref{eqn:13} with $M=0$: 
\begin{flalign}
\begin{split}
    \underbrace{H_{x} (j)x(t)+H_{v} (j)v_{n} }_{\hat{y}_{n} (j)}+\underbrace{H_{v} (j)}_{S_{y} (j)}\left(v-v_{n} \right) \\=H_{x} (j)x(t)+H_{v} (j)v \le 0, \ j=0,...,j^{*}.
\end{split}
\label{eqn:18}
\end{flalign}
Thus, for a linear system, MNN-RG with $M=0$ is equivalent to the standard linear RG. 
Note that the standard linear RG solves a linear programming problem, unlike the MNN-RG with $M\neq 0$, whose optimization may have quadratic constraints, as shown before.
\end{remark}
\begin{remark} 
For a single-input multi-output system, where \( y(t) \in \mathbb{R}^{n_y} \), \eqref{eqn:13} in MNN-RG can be replaced by the following \( n_y \) conditions:
\begin{flalign}
    \begin{split}
        &\hat{y}_{in}(j) + S_{yi} (j)(v - v_n) + \frac{M_i}{2} (v - v_n)^2 \leq 0,\\
        &i = 1, \dots, n_y, \ j=0,...,j^{*}.
    \end{split}
\label{eqn:19}
\end{flalign}
Here, \(\hat{y}_{in}(j)\) is the \(i\)-th nominal prediction trajectory of the constrained output \( y_i \), $S_{yi} (j)$ is the sensitivity of the \mbox{\(i\)-th} output with respect to \( v \), and \( M_i \ge 0 \) is a constant for the \(i\)-th output such that $\sup_{(x,v) \in (\mathcal{X} \times \mathcal{V})} \left| \frac{d^2 y_i}{dv^2} \right| \leq M_i
$. The resulting multi-output MNN-RG shares the same theoretical properties as the single-output one, as discussed in Theorem~1.
\end{remark}

\subsection{MNN-RG tuning}
\label{MNN-RG tuning}
In the previous section, we explained that the remainder term, $R_v(j)$, can be bounded as \(|R_v(j)| \leq \frac{M}{2} (v - v_n)^2\), where the parameter \(M\) can be calculated analytically as the smallest value satisfying ${\sup_{\substack{(x,v) \in (\mathcal{X} \times \mathcal{V})}} \left| \frac{d^2 y}{dv^2} \right| \leq M}$.
However, there are three challenges associated with analytical calculation of $M$:
\begin{enumerate}
    \item Such calculation requires the evaluation of the second-order sensitivity function of \(y\) with respect to \(v\), which can be quite tedious.
    \item Since such calculation is based on the upper bound of $|R_v(j)|$, the tracking performance of the MNN-RG can be compromised in cases where \(R_v(j) < 0\).
    \item The value of \(M\) may be too large, because it is based on the maximum mismatch between the nonlinear system output and its Taylor's approximation over the entire operating range (and not just in the vicinity of the constraint), which may lead to overly conservative response.
\end{enumerate}
 For these reasons, the analytical calculation of \(M\) for MNN-RG may be impractical and ineffective. In the following, we present two numerical approaches to ``tune'' the MNN-RG: 
 \begin{enumerate}
     \item The first approach, referred to as \(\bar{M}\)-tuning, involves replacing \(M\), which must be calculated analytically, with \(\bar{M}\), which is a calibration parameter determined numerically. In this approach, the upper bound of the remainder term \(\frac{M}{2}(v-v_n)^2\) in \eqref{eqn:14} is replaced with \(\frac{\bar{M}}{2}(v-v_n)^2\), where $\bar{M}$ is found empirically through an iterative process involving multiple hardware experiments or numerical simulations with the MNN-RG in the loop.
    \item The second approach, referred to as \(\bar{R}\)-tuning, involves bounding \(R_v(j)\) directly instead of tuning $M$. In this approach, the upper bound of the remainder term, ${\frac{M}{2}(v-v_n)^2}$ in \eqref{eqn:14}, is replaced with its approximated value referred to as \( \bar{R} \) obtained numerically by conducting a single experiment or simulation. 
\end{enumerate}
As we will show later, the tracking performance of \(\bar{R}\)-tuning is less favorable; however, its redeeming feature is ease of implementation, as it does not require an iterative process unlike \(\bar{M}\)-tuning.
In the following subsections, we provide a detailed explanation of each approach.

\subsubsection{$\bar{M}$-tuning}
\label{M-tuning}
In the $\bar{M}$-tuning method, we replace $M$ in \eqref{eqn:13}, \eqref{eqn:14} with a tunable parameter $\bar{M}$, that is: 
\begin{flalign}
\hat{y}_{n}(j) + S_{y}(j)(v - v_{n}) + \frac{\bar{M}}{2} (v - v_{n})^{2} \le 0, \; j = 0, \ldots, j^{*}.
\label{eqn:23}
\end{flalign}
The goal is to find the smallest value of $\bar{M}$ (to maximize the tracking performance) such that no constraint violations occur. The smallest value of $\bar{M}$ is found iteratively, starting from $\bar{M}=0$. For each $\bar{M}$, the nonlinear system is operated by running an experiment on the hardware or simulating the system if a simulator is available, using the MNN-RG approach. The output trajectory of the system is then checked for constraint violations. If any constraint violations are detected along the trajectory, $\bar{M}$ is increased so that the MNN-RG applies a less aggressive reference to the closed-loop plant. 
If, on the other hand, no constraint violations occur, $\bar{M}$ is decreased until they do or the RG becomes inactive (i.e., $\kappa=1$ in \eqref{eqn:23}). 

The above is summarized in Algorithm \ref{alg:M} for calibrating  $\bar{M}$.
The following proposition shows that, in Algorithm \ref{alg:M}, there exists a sufficiently large \(\bar{M}\) that guarantees the termination of the first while loop and a sufficiently large negative value of \(\bar{M}\) that ensures the termination of the second while loop.
\begin{algorithm}[t!]
\caption{Calibration of $\bar{M}$ in MNN-RG with $\bar{M}$-tuning}
\begin{algorithmic}[1]
\STATE \textbf{Inputs:} Reference signal $r(t)$, where $0 \leq t \leq T$, $\Delta \bar{M}$ (Step size based on the problem)
\STATE \textbf{Output:} $\bar{M}$ in MNN-RG approach with $\bar{M}$-tuning
\STATE Initialize $\bar{M} \leftarrow 0$
\STATE Execute the MNN-RG with $M$ replaced by $\bar{M}$ in (14) to obtain $y(t)$ for $0 \leq t \leq T$, using $r(t)$ as the reference
\IF{$y(t) > 0$ for any $0 \leq t \leq T$}
\WHILE{$y(t) > 0$ for any $0 \leq t \leq T$}
\STATE $\bar{M} \leftarrow \bar{M} + \Delta \bar{M}$
\STATE Execute the MNN-RG with $M$ replaced by $\bar{M}$ in (14) using $r(t)$ to obtain $y(t)$ for $0 \leq t \leq T$
\ENDWHILE
\ELSE
\WHILE{$y(t) \leq 0$ for all $0 \leq t \leq T$}
\STATE $\bar{M} \leftarrow \bar{M} - \Delta \bar{M}$
\STATE Execute the MNN-RG with $M$ replaced by $\bar{M}$ in (14) using $r(t)$ to obtain $y(t)$ and $\kappa(t)$ for $0 \leq t \leq T$
\IF{$y(t) > 0$ for any $0 \leq t \leq T$}
\STATE $\bar{M} \leftarrow \bar{M} + \Delta \bar{M}$
\STATE \textbf{break}
\IF{$\kappa(t) = 1$ for all $0 \leq t \leq T$}
\STATE \textbf{break}
\ENDIF
\ENDIF
\ENDWHILE
\ENDIF
\STATE \textbf{Output:} $\bar{M}$
\end{algorithmic}\label{alg:M}
\end{algorithm}

\begin{proposition}
Assume that $\exists M \in \mathbb{R}_{\geq 0}$ such that \(|R_v(j)| \leq \frac{M}{2} (v - v_n)^2\). Then, for any \(\bar{M} > M\), the left-hand side of \eqref{eqn:23} provides an upper bound on \(\hat{y}(j)\) from \eqref{eqn:9}, and for any \(\bar{M} < -M\), it provides a lower bound on \(\hat{y}(j)\).
\end{proposition}
\begin{proof}
The inequality \(|R_v(j)| \leq \frac{M}{2} (v - v_n)^2\) implies \(-\frac{M}{2} (v - v_n)^2 \leq R_v(j) \leq \frac{M}{2} (v - v_n)^2\). Substituting these bounds into \eqref{eqn:9} and comparing it with the left-hand side of \eqref{eqn:23} directly shows that for \(\bar{M} > M\), the left-hand side of \eqref{eqn:23} exceeds \(\hat{y}(j)\), and for \(\bar{M} < -M\), it is less than \(\hat{y}(j)\). 
\end{proof}
Proposition 1 implies that any \( \bar{M} > M \) breaks the first while loop, while any \( \bar{M} < -M \) breaks the second while loop.

\subsubsection{$\bar{R}$-tuning}
\label{bar{R}-tuning}
To describe the \(\bar{R}\)-tuning approach, we first rearrange the terms in (9) as:
\begin{flalign}
    \begin{split}
        R_v(j)=\hat{y}(j)-\hat{y}_{n} (j)-S_{y} (j)(v-v_{n} ).
    \end{split}
\label{eqn:20}
\end{flalign}
The goal is to find an upper bound on $R(j)$ using the above by conducting an experiment. To ensure that the constraints are satisfied during the tuning process, the experiment will be conducted with the PRG in the loop for a reference signal \(r(t)\) over \(0 \leq t \leq T\), where \(T\) represents both the experiment and the reference signal horizon. The upper bound on \(\bar{R}\) is numerically approximated using the collected data as follows:
\begin{flalign}
    \begin{split}
        \bar{R} = \max_{0 \leq t \leq T} \left( \max_{0 \leq j \leq j^*} {R_v}(t+j|t) \right)
    \end{split}
\label{eqn:21}
\end{flalign}
Here, we replace \( R_v(j) \) with the  notation \( R_v(t + j|t) \) to explicitly reference \( t \).
To explain \eqref{eqn:21}, we evaluate ${R}_v(t + j|t)$ over the prediction horizon $j$ (ranging from 0 to $j^*$) for each time step $t$ (ranging from 0 to $T$) and choose the maximum value. 
This leads to Algorithm \ref{alg:Rtuning} to bound the remainder term in \eqref{eqn:9}. 
With \(\bar{R}\) as the approximated upper bound of the remainder term, in MNN-RG with \(\bar{R}\)-tuning, \eqref{eqn:13} and thus the constraint in \eqref{eqn:14} are replaced with
\begin{flalign}
    \begin{split}
        \hat{y}_{n} (j) + S_{y} (j) (v - v_{n}) + \bar{R} \le 0,\ j=0,...,j^{*}.
    \end{split}
\label{eqn:22}
\end{flalign}
If the $\bar{R}$ is positive, it robustifies the MNN-RG by providing a conservative constraint tightening that accounts for the worst-case approximation error. Conversely, if $\bar{R}$ is negative, it allows for potential improvement in the tracking performance of the MNN-RG by relaxing the constraints.
\begin{algorithm}[t!]
\caption{Bounding the Remainder Term with \(\bar{R}\)-Tuning}
\label{alg:Rtuning}
\begin{algorithmic}[1]
\STATE \textbf{Input:} Reference signal \(r(t)\) for \(0 \leq t \leq T\)
\STATE \textbf{Output:} Upper bound of the remainder term, \(\bar{R}\)
\STATE Initialize \(\bar{R} \leftarrow -\infty\) (or the largest negative value)
\STATE Run the system with PRG using \(r(t)\), and collect state variables \(x(t)\) and input variables \(v(t)\) for \(0 \leq t \leq T\)
\FOR{$t = 0$ to $T$}
    \STATE Compute the output of the neural network, $v_n(t)$, with inputs \(x(t)\), \(v(t-1)\), and \(r(t)\)
    \FOR{$j = 0$ to $j^*$}
        \STATE Compute \( \hat{y}_n(j) \) with \eqref{eqn:8}, \( S_y(j) \) with \eqref{eqn:10}, and \( \hat{y}(j) \) with \eqref{eqn:8} where \( v_n \) is replaced with \( v(t) \).
        \STATE Calculate \(R_v(t + j|t)\) using (21)
        \STATE Update the upper bound: \(\bar{R} \leftarrow \max(\bar{R}, R_v(t + j|t))\)
    \ENDFOR
\ENDFOR
\STATE \textbf{Output:} \(\bar{R}\)
\end{algorithmic}
\end{algorithm}

\subsection{The entire MNN-RG design process}
\label{MNN-RG design process}
In Fig.~\ref{fig:2}, a flowchart for the training and performance evaluation of the proposed MNN-RG with $\bar{M}$-tuning is shown. The trend is similar for the MNN-RG with $\bar{R}$-tuning. The flowchart can be broken down into the following main steps:
\vspace{0.3em}

\noindent \textbf{Step 1: Data collection and pre-processing:}

\noindent To collect input and output features for training a regression NN, the PRG strategy is implemented on the nonlinear system for a given reference signal $r(t)$ by conducting experiments on the hardware or simulating the system if a simulator is available. To enhance the prediction accuracy of the NN model, the initial dataset is normalized before it is used for NN training \cite{geron2022hands}.

\noindent \textbf{Step 2: Optimal neural network selection:}

\noindent After data collection, several NNs are trained in multiple trials, each with distinct configurations, including variations in the number of hidden layers and the number of neurons in those layers. In these trials, the initial values for biases, weights, and data splitting are randomized. After evaluating the results from all trials, the trial that yields the highest performance across various metrics, such as root mean square error (RMSE), R-square, or others, is selected for use in the  MNN-RG algorithm. For more details on deriving suitable NN topologies for approximating the PRG, please refer to \cite{lim2023reference}.

\vspace{0.3em}

\noindent \textbf{Step 3: MNN-RG tuning and performance check:} 

\noindent We select a reference signal that is different from the one used in NN training  (including the validation and test datasets). For this reference, parameter $\bar{M}$ is adjusted using Algorithm \ref{alg:M}. Following this, the performance of the MNN-RG  approach with $\bar{M}$-tuning is assessed by comparing its output with that of the PRG using metrics such as R-square, RMSE, or others. If the MNN-RG fails to achieve satisfactory reference tracking while maintaining constraint enforcement, it indicates that the NN lacks generalization to unseen data. The poor generalization can be attributed to either an insufficient amount of training data for the NN or challenges such as  presence of noise and outliers in the collected data, and inherent model limitations. By analyzing the NN's performance during the training phase, we can determine whether additional data collection is necessary or if a more sophisticated machine learning model and advanced data preprocessing techniques should be employed to achieve the desired tracking performance. Achieving any arbitrary desired tracking performance with the the MNN-RG approach is impractical due to limitations such as finite training data, model constraints, and computational resources. So, it is important to  adapt the threshold for evaluating the performance of the MNN-RG tracking to a practical or achievable level. 
\begin{remark}
The command profile $r$ used for dataset collection during NN training should differ from the reference signal employed for the optimal selection of $\bar{M}$ in Algorithm \ref{alg:M}. This is because the requirement for robustification or tracking performance of the MNN-RG using the $\bar{M}$ term is attributed to the NN's limited generalization ability to unseen data. Utilizing the same command profile to tune $\bar{M}$ defeats this purpose. Tuning the parameter $\bar{M}$ based on a closed-loop run, where the system faces diverse operating conditions and (small) constraint violations, increases the likelihood that an MNN-RG with the tuned $\bar{M}$ will exhibit good tracking performance without constraint violations in new scenarios.
\end{remark}
\begin{figure}
\centering
\includegraphics[width=0.5\textwidth]{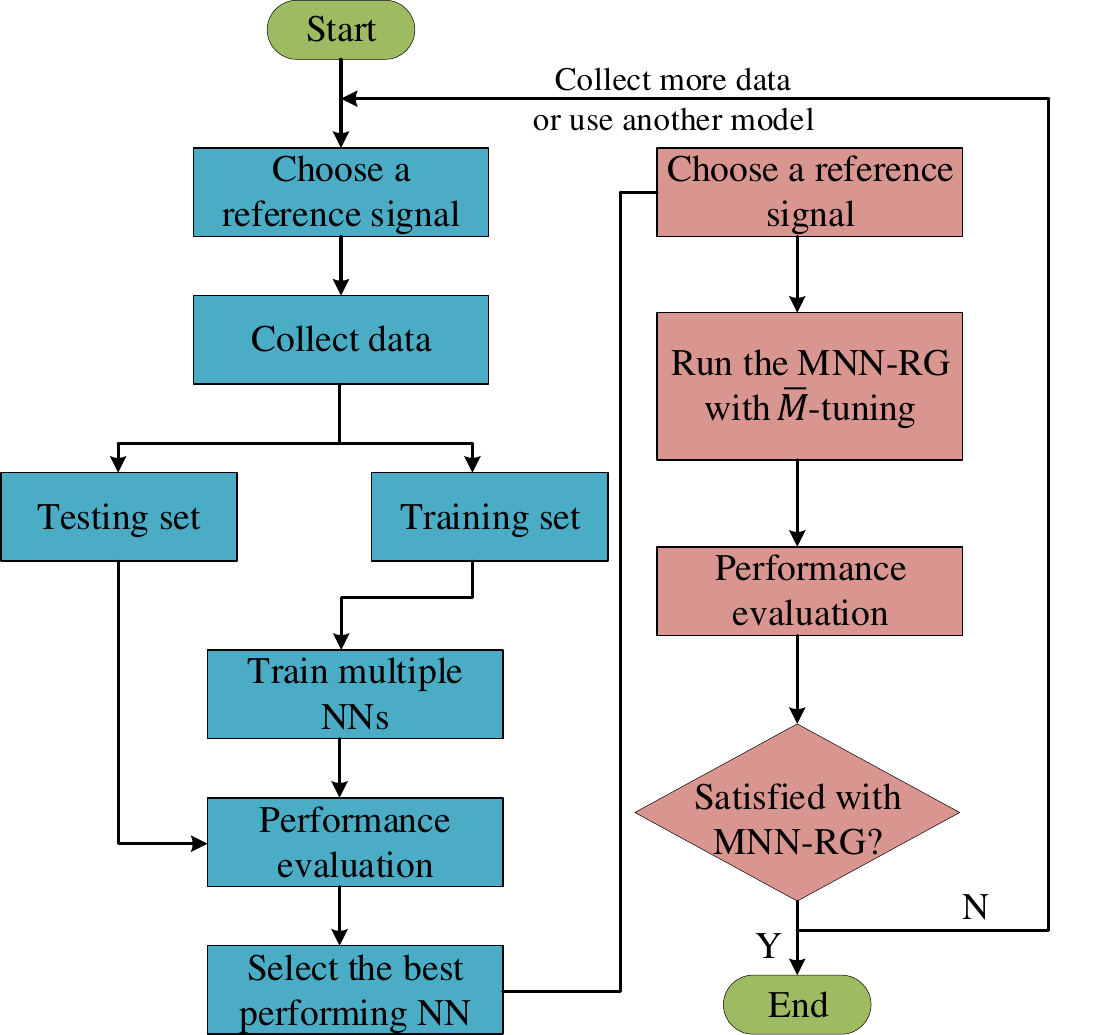}
\caption{Flowchart for the offline training and performance evaluation of the MNN-RG approach with $\bar{M}$-tuning. The blue boxes correspond to the neural network training and evaluation steps. The red boxes correspond to MNN-RG tuning and evaluation steps.}
\label{fig:2}
\end{figure}

\section{MNN-RG for constraint management of the Fuel Cell (FC) system}
\label{Implementation of the modified NNRG on the FC air-path system}
In this section, the proposed MNN-RG approach is utilized for  constraint management of the FC air-path subsystem. The plant model used in this study is a nonlinear, spatially averaged, reduced third-order model of a FC stack with its auxiliaries \cite{talj2009experimental}. This low-order model is derived from a nine-state full-order system developed in \cite{pukrushpan2004control}, and is suitable for studying the air-path side of the FC system. Below, we first introduce the closed-loop FC airpath system, i.e., system \eqref{eqn:1}, and show that without a reference governor, system constraints are violated, which motivates the need for the MNN-RG. In Subsection~\ref{NN-RG for load control}, we show that an NN-RG strategy without modification cannot enforce the constraints. Therefore, in Subsection~\ref{Implementation of the MN-RG}, we tune and apply the MNN-RG to the closed-loop FC system and show that it can successfully enforce the constraints. In Section~\ref{Robustness}, we conduct an additional test using a higher-fidelity, nine-state full-order model to evaluate the robustness of the proposed control strategy against model mismatch, noise, and disturbances. In Subsection~\ref{Drive-cycle simulation}, we present simulation results on a realistic drive cycle, and present an analysis of MNN-RG online execution time.

\subsection{Air-path system model for FC}
\label{Air-path system model for FC}
The dynamics of the FC air-path system can be described by three dynamic states $x={\left[p_{ca},{\omega }_{cp},p_{sm}\right]}^\top$, consisting of the cathode air pressure, the angular velocity of the compressor motor, and the pressure of air in the supply manifold, respectively. The state-space representation of a reduced third-order FC model is presented in \cite{talj2009experimental}. 

Control of the air-path system aims to maximize net power while preventing oxygen starvation and avoiding surge and choke violations across a range of load conditions \cite{vahidi2006constraint}. In the air-path system, the control input is the compressor motor voltage, $v_{cm}$, which is computed by the controller as a combination of static feedforward (SFF) control and proportional-plus-integral (PI) feedback control, as follows:
\begin{flalign}
    \begin{split}
        v_{cm} &=\underbrace{{\rm g_{1}I}_{st} {\rm +g_{2}}}_{SFF}{\rm +}\underbrace{\left(k_{p}+\frac{k_{I}}{s} \right)\left(W_{cp}^{ref} -W_{cp} \right)}_{PI}\\ W_{cp}^{ref} &=n_{cell} M_{O_{2} } I_{st} \lambda _{O_{2} ,des} \frac{1+\omega _{atm} }{4Fx_{O_{2} ,atm} }
    \end{split}
\label{eqn:24}
\end{flalign}
where $I_{st}$ [A] is the current drawn from the FC stack, \mbox{$W_{cp}$ [kg s${}^{-1}$]} denotes the compressor mass flow rate, $n_{cell}$ is the number of cells in the stack, $F=96485$ [C mol${}^{-1}$] is the Faraday number, $M_{O_2}=32$ [g mol${}^{-1}$] is the mole mass of oxygen, \mbox{${\lambda }_{O_2,des}=2$ [-]} is the desired value of oxygen excess ratio to maximize the net power of the stack, \mbox{${\omega }_{atm}=0.0098$ [-]} is the humidity ratio of the \mbox{atmospheric} air, \mbox{$x_{O_2,atm}=0.233$ [-]} is the \mbox{atmospheric} oxygen mass fraction. The PI and static feedforward gains are provided in Appendix A. The details of the feedforward and feedback controller design for the air-path system are given in \cite{pukrushpan2004control}.

The constrained outputs of the air-path subsystem are ${y={\left[{\lambda }_{O_2},W_{cp},\ p_{sm}\right]}^\top}$, which we use to define the surge, choke, and oxygen starvation constraints respectively as:
\begin{flalign}
     \begin{split}
        \frac{p_{sm} }{p_{atm} }& \le 50W_{cp} -0.1\\ 
        \frac{p_{sm} }{p_{atm} } &\ge 15.27W_{cp} +0.6\\ 
        \lambda _{O_2} &\ge 1.9
    \end{split}
\label{eqn:25}
\end{flalign}
where
\begin{flalign*}
    \begin{split}
        \lambda _{O_{2} } &=\frac{W_{O_{2} ,ca,in} }{W_{O_{2} ,ca,rct} }\\ 
    \end{split}
\end{flalign*}
\begin{flalign*}
    \begin{split}
        W_{O_{2} ,ca,in} &=k_{ca,in} \frac{x_{O_{2} ,atm} }{1+\omega _{atm} } \left(p_{sm} -p_{ca} \right)\\ 
        W_{O_{2} ,ca,rct} &=\frac{n_{cell} M_{O_{2} } }{4F} I_{st}
    \end{split}
\end{flalign*}
in which $W_{O_2,ca,in}$ is the oxygen mass flow rate into the cathode, $W_{O_2,ca,rct}$ is the mass flow rate of the oxygen reacted in the cathode, $p_{atm}$ is the atmospheric pressure, and $k_{ca,in}=3.62\times 10^{-6} \, [\text{kg}\,\text{Pa}\,\text{s}^{-1}]$ is the constant for the cathode inlet nozzle~\cite{vahidi2006constraint}. 

Fig.~\ref{fig:3} illustrates the OER and the compressor trajectory responses of the closed-loop model (without any constraint management strategy) during a series of step-up and step-down changes in stack current $I_{st}$. 
The figure reveals that the surge constraint may be violated during current step-down, while the OER and choke constraints may be violated during current step-up. Since fast changes in current result in constraint violations, we employ a current governor (also known as a load governor) strategy to enforce constraints on the OER and compressor.
\begin{figure}
\centering
\includegraphics[width=0.5\textwidth]{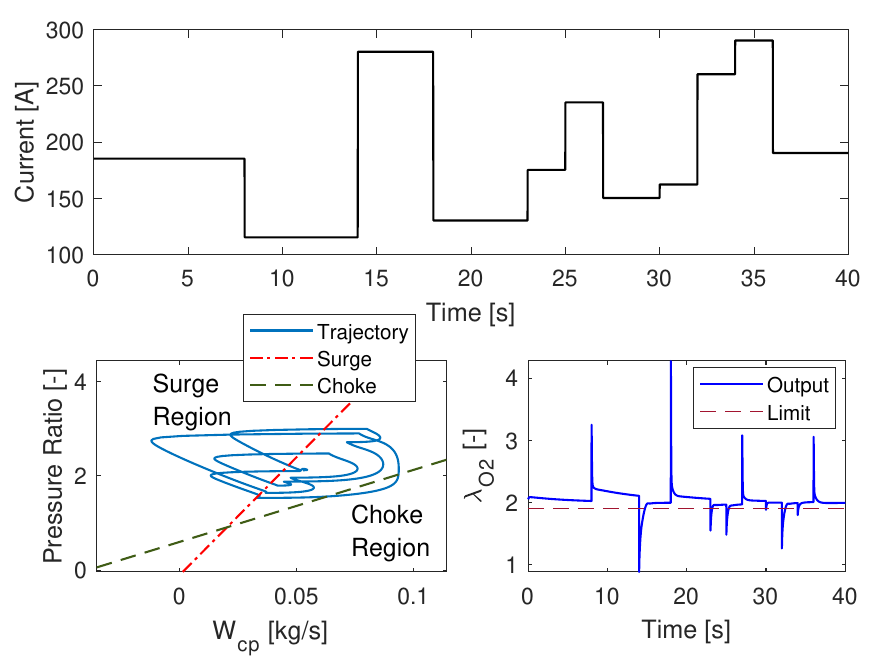}
\caption{Simulation response without reference governor: Current, OER, and compressor state trajectories for step changes in current demand. The dashed lines indicate system constraints. The surge region lies to the left of the dashed red line (surge boundary) and the choke region lies to the right of the dashed green line (the choke boundary). The region below $\lambda _{O_{2} } = 1.9$ on the OER plot indicates oxygen starvation and should be avoided for durable operation of  FC.}
\label{fig:3}
\end{figure}
\begin{figure}[b!]
\centering
\includegraphics[width=0.5\textwidth]{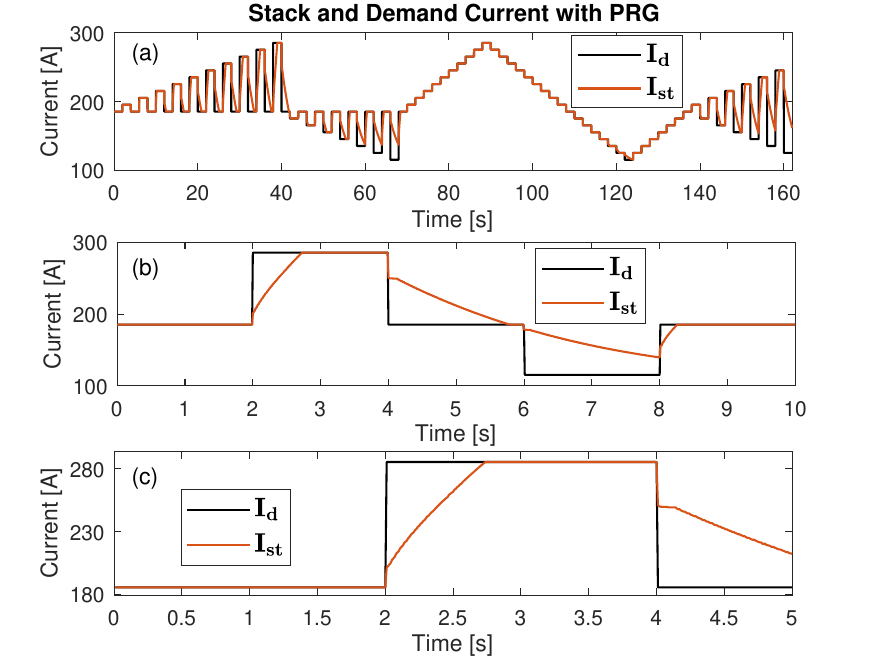}
\caption{Simulation results of PRG on various reference profiles for data collection: (a) comprehensive dataset with $N=16200$ data points, (b) partial dataset with $N=1000$ data points, and (c) limited dataset with $N=500$ data points.}
\label{fig:4}
\end{figure}

\subsection{Data collection, training, and NN-RG for load control}
\label{NN-RG for load control}
In this section, the NN-RG approach described previously in section~\ref{Neural network reference governor (NN-RG)} is employed as the load governor strategy. Recall that in NN-RG, a collection of input and output data from the closed-loop system controlled by the PRG is required in order to train the NN. Here, the inputs to the PRG (and therefore the NN) are the  state $x(t)$ (of both the airpath and the controller), the previous modified reference $v(t-1)=I_{st}(t-1)$, and the desired reference $r\left(t\right)=I_d(t)$. The output of the PRG is the modified reference $v\left(t\right)=I_{st}(t)$. 

To understand the impact of training dataset size, i.e., the number of training samples used for NN training, we use three different datasets that are obtained by simulating the closed-loop system with the PRG (i.e., Algorithm \ref{alg:PRG}). Each of the three simulations uses a different reference profile, as shown in Fig.~\ref{fig:4}. For these simulations, the sampling frequency is chosen as $100$ Hz, and the output prediction horizon is chosen to be $5$ s, the same value used in \cite{sun2005load}. This corresponds to $j^*=500$. The number of iterations in Algorithm \ref{alg:PRG} is set to $L=15$, and $\varepsilon$, the steady-state tightening, is set to 0.05. The datasets can be described as follows: a comprehensive dataset with $N=16200$ data points (as shown in Fig.~\ref{fig:4}a), where the desired reference covers a diverse range of operating conditions; a partial dataset with only $N=1000$ data points (as shown in Fig.~\ref{fig:4}b), where the desired reference covers the entire current range through two large step changes in the current; and a limited dataset with $N=500$ data points (as shown in Fig.~\ref{fig:4}c), where the desired reference covers only a part of the current demand range.

For each NN, we use a feedforward architecture, configured with a single hidden layer (HL) that employs the hyperbolic tangent sigmoid function as its activation function and a linear function in the output layer.  
Training of the NN is performed using the collected data, utilizing the command \ttfamily{fitnet} \normalfont of MATLAB. The dataset is randomly divided into three sets with 70\% of the dataset designated for training, and the remaining data equally split between validation and testing, each comprising 15\%. 
The cost function, the training algorithm, and the stopping criteria of the algorithm are set to the recommended default values of {\ttfamily{fitnet}}. 
Table~\ref{tab:Structure and performance evaluation of the NNs} lists the number of neurons used in the hidden layer of each NN along with the corresponding RMSE between the outputs of the NN and the stack current generated from PRG, confirming the successful training of each NN.

To evaluate the generalization ability of the three neural network models, we simulate the closed-loop system using the PRG with a reference signal not encountered during the NN training. The simulation data (namely, $x(t)$, $r(t)$, and $v(t-1)$) are then used as inputs to the NNs, and their outputs are evaluated and compared against the PRG. The results are presented in Fig.~\ref{fig:6}.
By comparing the NNs' outputs to the stack current generated by the PRG, it can be seen that the NN models trained on the comprehensive, partial, and limited datasets exhibit high, moderate, and low prediction accuracies, respectively.
\begin{figure}
\centering
\includegraphics[width=0.4\textwidth]{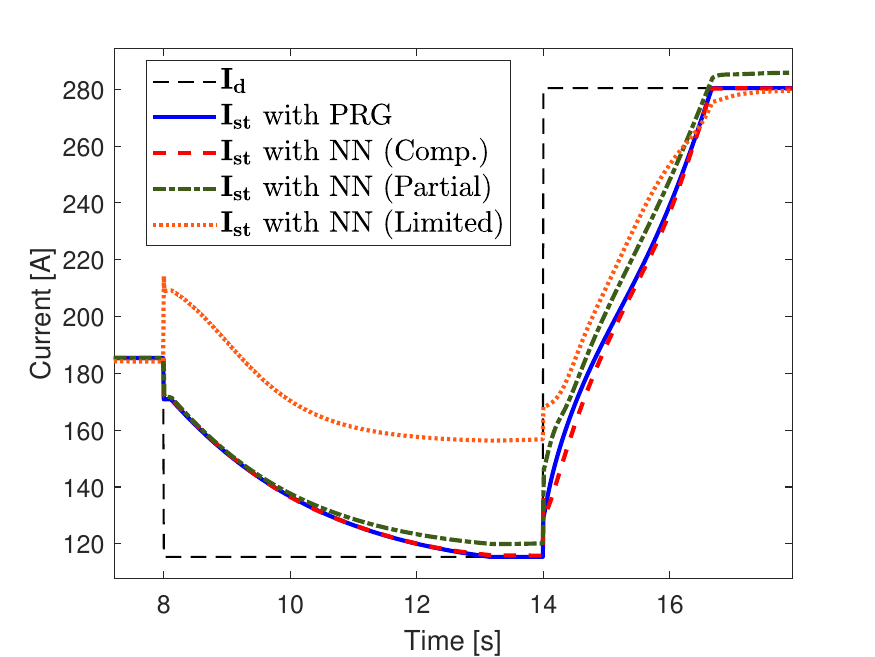}
\caption{Comparing the prediction accuracies of the NNs trained on the comprehensive, partial, and limited datasets.}
\label{fig:6}
\end{figure}
\begin{table}[b!]
\centering
\caption{\\ \uppercase{Structure and Performance Evaluation of The Neural Networks.}}
\small
\begin{tabular}{|l|c|c|}
\hline
\multicolumn{1}{|c|}{Case}                  & HL Neurons & \multicolumn{1}{l|}{RMSE [Amps]} \\ \hline
\multicolumn{1}{|c|}{Comprehensive dataset} & $10$       & $1.0311$                    \\ \hline
Partial dataset                             & $2$        & $2.5232$                    \\ \hline
Limited dataset                             & $5$        & $1.9598$                    \\ \hline
\end{tabular}
\label{tab:Structure and performance evaluation of the NNs}
\begin{tablenotes}
\item \hspace{-1em} \textbf{Note}: The RMSEs are evaluated on the datasets used during the training process, i.e., training, testing, and validation data.
\end{tablenotes}
\end{table}
Before introducing the MNN-RG for load control of the FC in the next section, we initially utilize the trained NNs to enforce the OER and compressor constraints in FC system, as shown by the block diagram in Fig. ~\ref{fig:5}. 
\begin{figure}[b!]
\centering
\includegraphics[width=1.03\columnwidth]{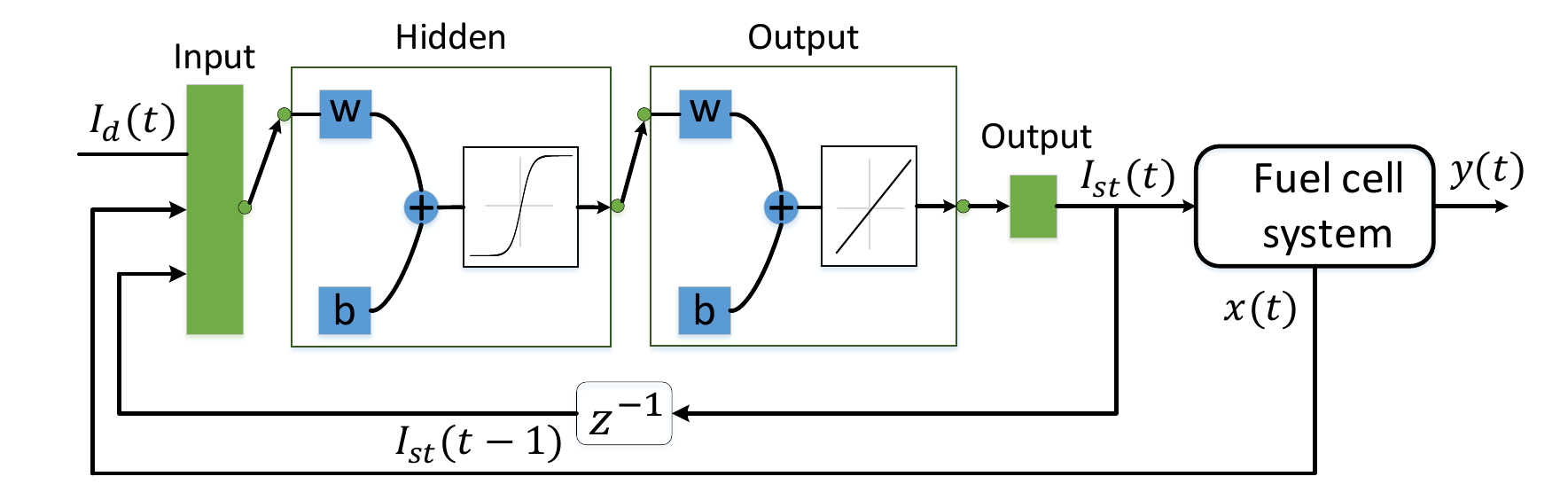}
\caption{Schematics of the NN-RG applied to the FC air-path system.}
\label{fig:5}
\end{figure} 
To this end, a reference profile covering diverse FC system operating conditions is designed. The results are shown in Fig.~\ref{fig:7} where, as expected, the well-generalized NN results in better performance as compared to the NN with poor generalization, and they all result in some level of constraint violation, which we are going to address with MNN-RG next.
\begin{figure*}
    \centering
    \resizebox{\textwidth}{0.25\textheight}{%
        \includegraphics{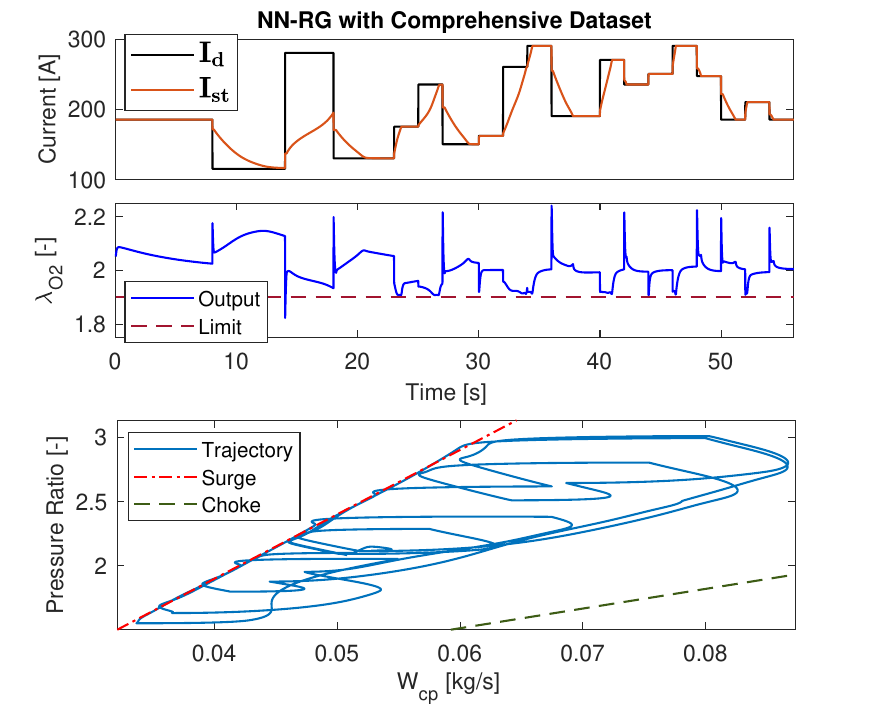}%
        \hspace{-1.3cm}%
        \includegraphics{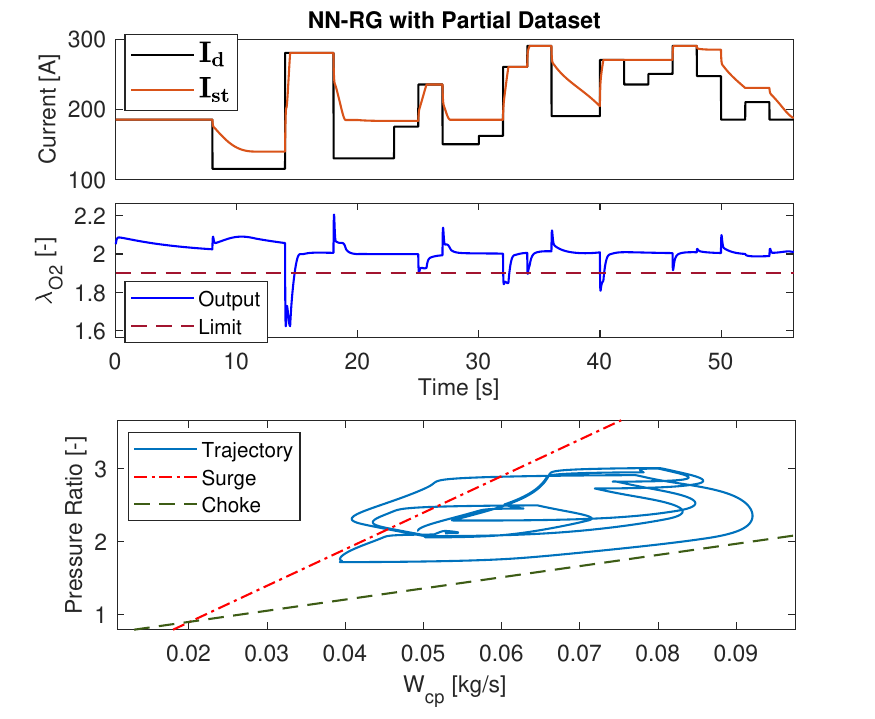}%
        \hspace{-1.3cm}%
        \includegraphics{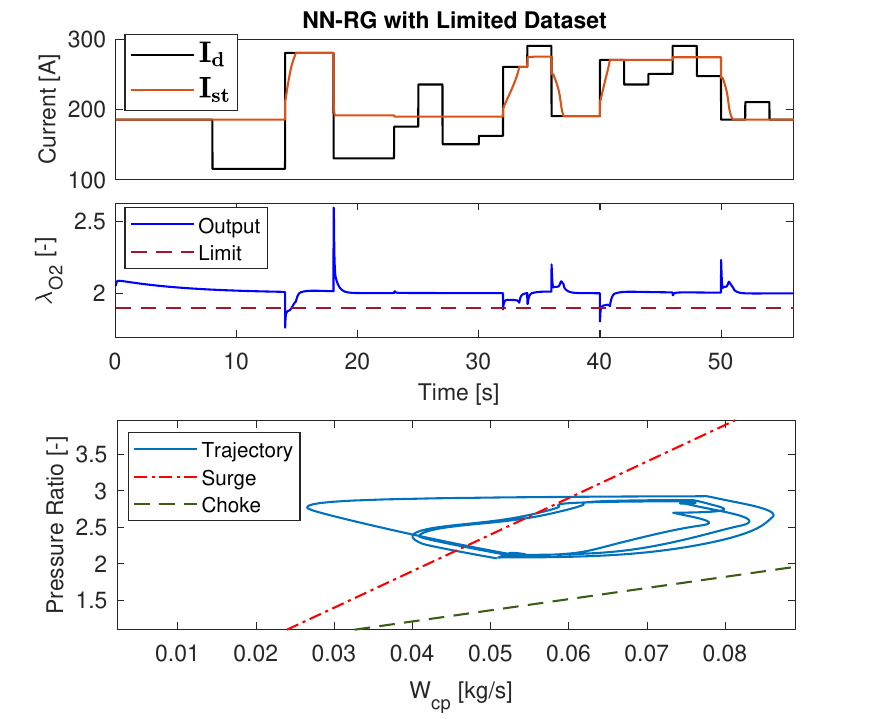}%
    }
    \caption{Performance of the NN-RGs trained on the comprehensive, partial, and limited datasets for a random step series of reference profile.}
    \label{fig:7}
\end{figure*}

\subsection{MNN-RG for load control}
\label{Implementation of the MN-RG}
In this section, the MNN-RG algorithm is employed as load governor for constraint enforcement in FC system. The MNN-RG uses the same trained NNs from the previous subsection.

Since the FC system has multiple constrained outputs, we utilize \eqref{eqn:19} and \eqref{eqn:25} to re-write \eqref{eqn:13} for the OER and surge constraints as follows:
\begin{flalign}
    \begin{split}
        &\frac{M_{1}}{2} (v-v_{n})^{2} - S_{y_{1}}(j)(v-v_{n}) - \hat{y}_{1n}(j) + 1.9 \leq 0 \\
        &\frac{\left(M_{3} + 50p_{\text{atm}} M_{2}\right)}{2}(v-v_{n})^{2} + \left(S_{y_{3}}(j) - 50p_{\text{atm}} S_{y_{2}}(j)\right) \\
        &\quad \times (v-v_{n}) + \hat{y}_{3n}(j) - 50p_{\text{atm}} \hat{y}_{2n}(j) + 0.1p_{\text{atm}} \leq 0
    \end{split}
\label{eqn:26}
\end{flalign}
where $y_1 = {\lambda }_{O_2}, y_2 = W_{cp}$, and $y_3 = p_{sm}$ as before, and $M_{i} \ge 0$, $i=1,2,3$ is a constant for the \(i\)-th output of the FC system such that \( \left| \frac{d^2 y_i}{dv^2} \right| \leq M_i \) for all $x$ and $v$ within the operating range of the FC system. The fuel cell model satisfies the required assumptions in Section~\ref{Problem setup}; therefore, the existence of the $M$ term is ensured in this case. The sensitivities of the constrained outputs with respect to the reference, \( S_{y_i}, i = 1, 2, 3 \), along with the nominal state equations and state sensitivity equations for the air-path system, are also calculated and presented in Appendix B and C. Note that \eqref{eqn:26} excludes the choke constraint, as for the given model, satisfying the OER constraint drives the compressor trajectory away from its choke boundary. Additionally, since the OER constraint is imposed as a lower bound, the first inequality uses the lower bound of the remainder term, as opposed to the upper bound in \eqref{eqn:14}. The final step in MNN-RG involves substituting \eqref{eqn:26} into \eqref{eqn:14} and computing the optimal stack current. However, as discussed in Section~\ref{MNN-RG tuning}, the analytical computation of $M_{i}$ for MNN-RG may be impractical and ineffective. Therefore, in the following section, we explore and compare the two numerical approaches, $\bar{M}$ and $\bar{R}$-tuning, introduced in Section~\ref{M-tuning} and~\ref{bar{R}-tuning}, respectively, for the implementation of the MNN-RG controller.

\subsubsection{Comparison of $\bar{M}$ and $\bar{R}$-tuning methods} For MNN-RG with $\bar{M}$-tuning, the $M_{1}$ in the OER constraint and $M_{3} + 50p_{\text{atm}} M_{2}$ in the surge constraint in \eqref{eqn:26} are respectively substituted with $\bar{M}_{\lambda_{O_{2}}}$ [$\mathrm{A}^{-2}$] and $\bar{M}_{\mathrm{s}}$ [Pa $\mathrm{A}^{-2}$], which are the calibration parameters as explained in Section~\ref{M-tuning}. These parameters are then fine-tuned through Algorithm \ref{alg:M} with $\Delta \bar{M}=1$ for $\bar{M}_s$ and $\Delta \bar{M}=10^{-5}$ for $\bar{M}_{\lambda_{O_2}}$. For the $\bar{R}$-tuning, the quadratic terms of the OER and surge constraints in \eqref{eqn:26} are respectively replaced with $\bar{R}_{\lambda_{O_{2}}}$ [-] and $\bar{R}_{\mathrm{s}}$ [Pa], calculated as follows:
\begin{equation}
\begin{split}
&\bar{R}_{\lambda_{O_{2}}} = \max_{0 \leq t \leq T} \left( \max_{0 \leq j \leq j^*} -{R}_{1v}(t+j|t) \right)\\
&\bar{R}_{s} = \max_{0 \leq t \leq T} \left( \max_{0 \leq j \leq j^*} {R}_{3v}(t+j|t)-50p_{atm}{R}_{2v}(t+j|t) \right)
\end{split}
\label{eqn:27}
\end{equation}
where ${R}_{iv}(t+j|t)$ represents \eqref{eqn:20} for the $i$-th output of the air-path system. The negative sign in $\bar{R}_{\lambda_{O_{2}}}$ equation
is due to the OER constraint being imposed as a lower bound.

To compare the impact of the $\bar{M}$ and $\bar{R}$-tuning approaches on the tracking performance and constraint satisfaction of MNN-RG, we conduct two distinct case studies: one based on NN trained on the comprehensive dataset and one based on the NN trained on the partial dataset. For both cases, the reference trajectory from Fig.~\ref{fig:7} is utilized as the current demand, and the parameters for $\bar{M}$ and $\bar{R}$-tuning are computed as described above, see Table~\ref{tab:the parameters and tuning methods}.
\begin{table}[b]
\centering
\caption{\\ \uppercase{The Parameters in The $\bar{M}$ and $\bar{R}$-Tuning Methods}}
\small
\begin{tabular}{|c|c|c|c|c|}
\hline
\textbf{Case}            & \textbf{$\bar{M}_{\lambda_{O_{2}}}$}              & \textbf{$\bar{M}_{\mathrm{s}}$}  & \textbf{$\bar{R}_{\lambda_{O_{2}}}$}  & \textbf{$\bar{R}_{\mathrm{s}}$}  \\ \hline
Comp. dataset   & $-1.6\times 10^{-4}$ & $26$                        & $0.0014$                              & $2575.9$                          \\ \hline
Partial dataset & $-0.9\times 10^{-4}$ & $26 $                       & $0.026 $                             & $18731 $                          \\ \hline
\end{tabular}
\label{tab:the parameters and tuning methods}
\end{table}
Fig.~\ref{fig:8} and Fig.~\ref{fig:9}  show the closed-loop responses of the OER and compressor trajectory using (a) the $\bar{M}$-tuning and (b) the $\bar{R}$-tuning. To quantify the conservativeness of the two MNN-RGs (namely, with $\bar{M}$ and $\bar{R}$-tuning), the RMSE between stack currents of the PRG and the MNN-RG are calculated and listed in Table~\ref{tab:RMSE Comparison}.
\begin{table}[b]
\centering
\caption{\\ \uppercase{Performance of the MNN-RG with $\bar{M}$ and $\bar{R}$-Tuning}}
\small
\begin{tabular}{|c|c|c|}
\hline
\textbf{Case}              & $\bar{M}$-Tuning & $\bar{R}$-Tuning \\ \hline
Comp. dataset   & $0.055$      & $2.0829$    \\ \hline
Partial dataset & $0.196$      & $14.8473$    \\ \hline
\end{tabular}
\label{tab:RMSE Comparison}
\end{table}
The following points can be concluded from this table and Figs.~\ref{fig:8} and~\ref{fig:9}:
\begin{figure}
    \centering
   \begin{subfigure}{0.45\textwidth}
        \centering        \includegraphics[width=\textwidth]{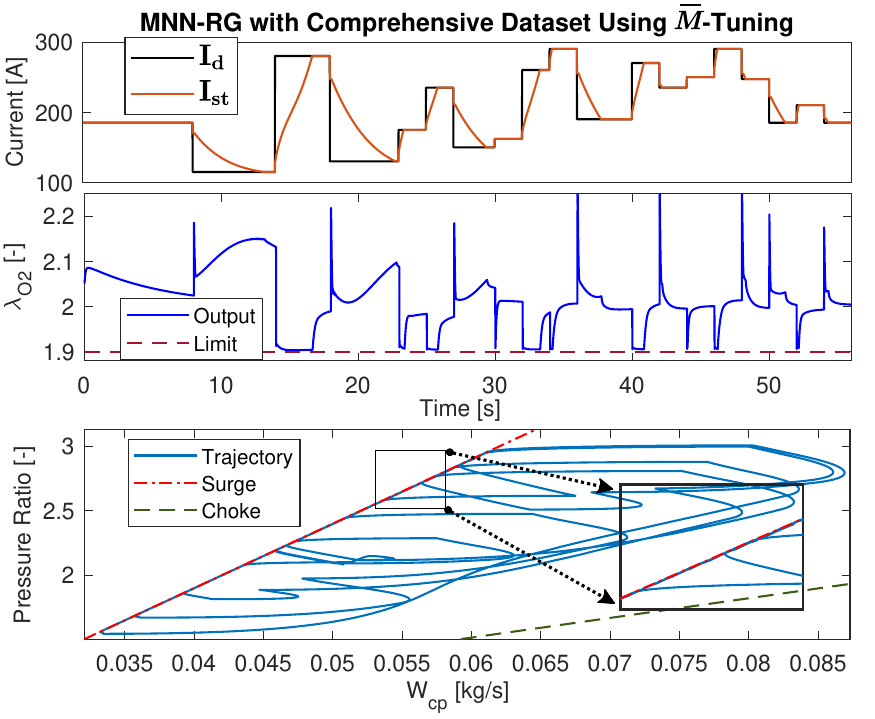}
        \caption{$\bar{M}$-tuning}
    \end{subfigure}
    \begin{subfigure}{0.45\textwidth}
        \centering
        \vspace{1.2em} 
    \includegraphics[width=\textwidth]{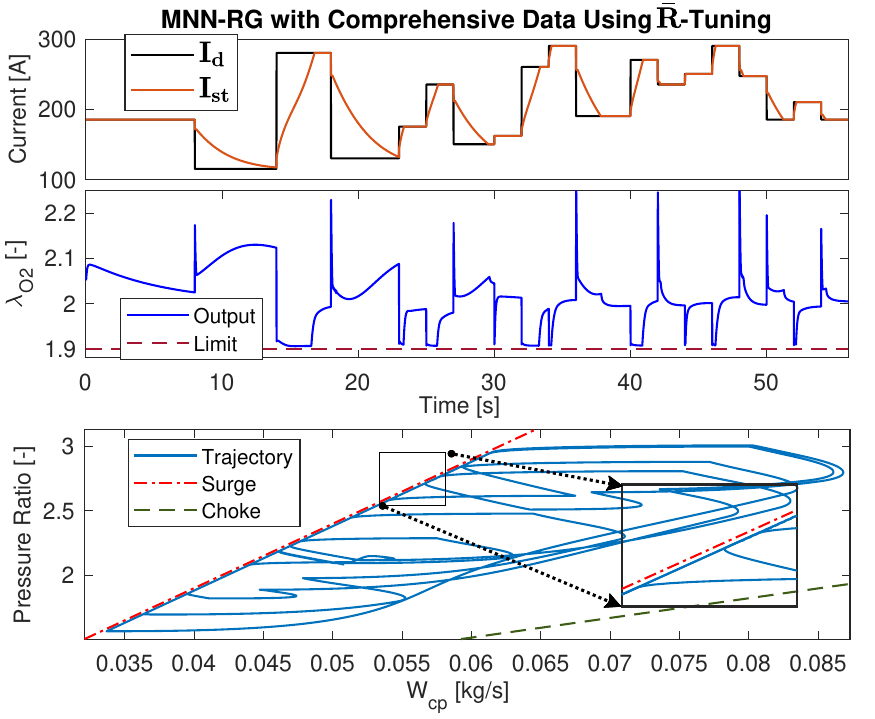}
        \caption{$\bar{R}$-tuning}
    \end{subfigure} 
    \caption{Simulation of MNN-RG trained on the comprehensive dataset}
    \label{fig:8}
\end{figure}
\begin{itemize}
    \item MNN-RG with $\bar{M}$ and $\bar{R}$-tuning satisfy the air-path constraints; however, in both methods, the current response from the MNN-RG trained on the partial dataset is more conservative than the one trained on the comprehensive dataset. This underscores the dependence of the MNN-RG response properties on the generalization capability of the utilized NN. 

    \item The OER and compressor response of the MNN-RG with the $\bar{R}$-tuning approach is more conservative compared to the MNN-RG with the $\bar{M}$-tuning. We can also point out that the tracking performance of the MNN-RG with $\bar{R}$-tuning tends to become more conservative and notably worse compared to the $\bar{M}$-tuning as the generalization of the NN within the MNN-RG degrades. However, as pointed out before, the redeeming feature of $\bar{R}$-tuning is its ease of implementation.
\end{itemize}
\begin{figure}
    \centering
   \begin{subfigure}{0.45\textwidth}
        \centering
        \includegraphics[width=\textwidth]{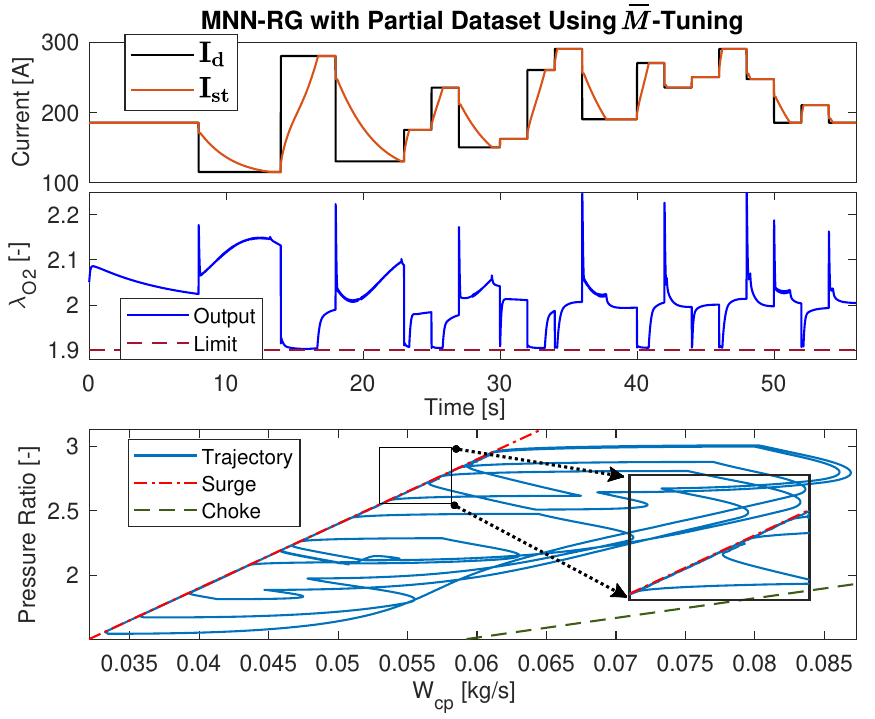}
        \caption{$\bar{M}$-tuning}
    \end{subfigure}
    \begin{subfigure}{0.45\textwidth}
        \centering
        \vspace{1.2em} 
        \includegraphics[width=\textwidth]{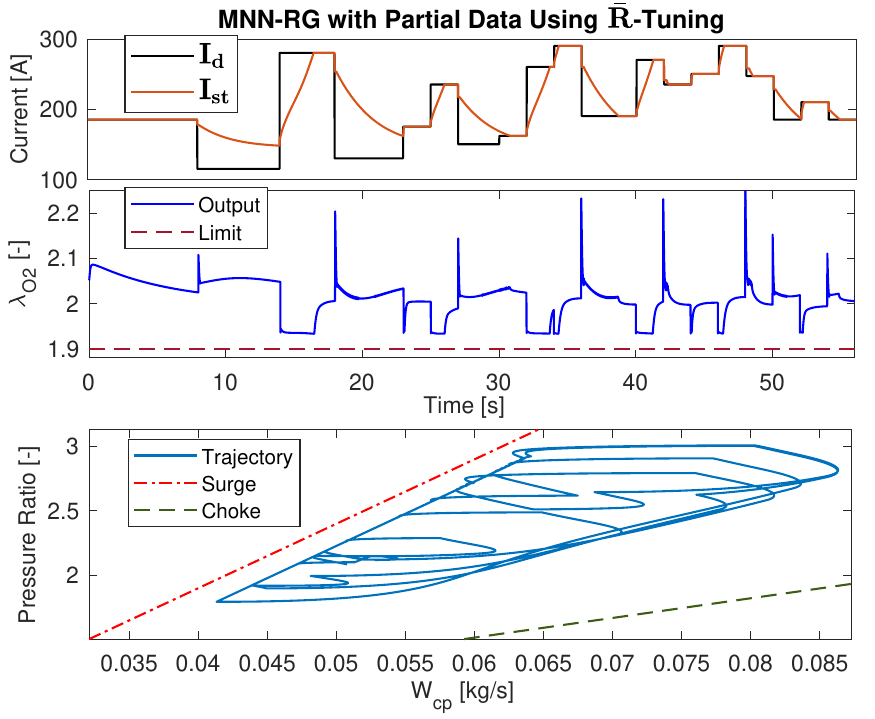}
        \caption{$\bar{R}$-tuning}
    \end{subfigure} 
    \caption{Simulation results of MNN-RG trained on the partial dataset}
    \label{fig:9}
\end{figure}
It is remarkable that the MNN-RG, with only two neurons trained on 1000 data points using the $\bar{M}$-tuning approach, satisfies all the constraints and matches the PRG in load tracking performance, even though the well-generalized NN-RG, trained with 10 neurons on 16200 data points, does not achieve this across the full operating range, as shown in Fig.~\ref{fig:7}(a). One could argue that the good generalization of the NN within the MNN-RG may not be necessary and the MNN-RG with $\bar{M}$-tuning can always achieve a fast tracking response if the $\bar{M}$ parameter is tuned appropriately.  In response to this argument, the following section examines the deficiencies of the MNN-RG in instances where the NN has poor generalization capabilities.

\subsubsection{The impact of poor NN generalization on the MNN-RG performance} 
Following a similar approach as the previous section, here we utilize the poorly generalized NN from Subsection~\ref{NN-RG for load control} for the MNN-RG with $\bar{M}$-tuning. Algorithm~\ref{alg:M} is used to calibrate the parameters in $\bar{M}$-tuning, resulting in $\bar{M}_{\lambda_{O_{2}}}=-0.7 \times 10^{-4}$ and $\bar{M}_{\mathrm{s}}=24$. Fig.~\ref{fig:10} shows the simulation results of the OER and compressor trajectory of the MNN-RG with a poorly generalized NN.
\begin{figure}
\centering
\begin{subfigure}{0.45\textwidth}
        \centering     \includegraphics[width=\textwidth]{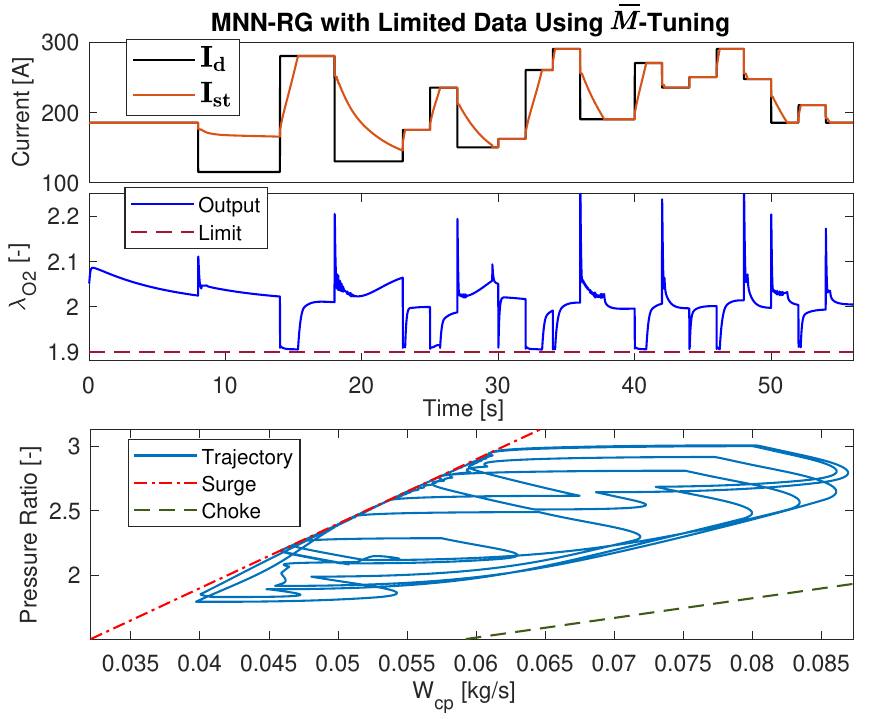}
    \end{subfigure}
\caption{ Simulation result of the MNN-RG trained on the limited dataset using the $\bar{M}$-tuning approach}
\label{fig:10}
\end{figure}
Comparing Figs.~\ref{fig:8}(a) and~\ref{fig:9}(a) against Fig.~\ref{fig:10}, it is evident that for the step current drops at \mbox{$t = 8$ s} and $t = 18$ s, the tracking performance of the MNN-RG with the poorly generalized NN is compromised to satisfy the surge constraint. This is because, as shown in  Fig.~\ref{fig:11}, with $\bar{M}_{\mathrm{s}}=0$, the MNN-RG simultaneously violates the compressor’s surge constraint and fails to maintain current tracking. Algorithm \ref{alg:M} increases $\bar{M}_{\mathrm{s}}$ to meet the surge constraint, which further worsens the current tracking performance.
\begin{figure}
\centering
\includegraphics[width=\columnwidth]{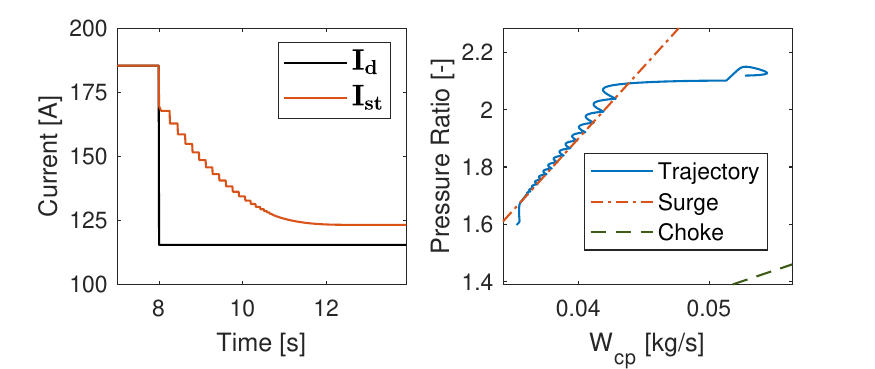}
\caption{ Simulation result of MNN-RG trained on the limited dataset with $\bar{M}_{\mathrm{s}}=0$}
\label{fig:11}
\end{figure}
This shows that despite the fine-tuning of the parameter $\bar{M}$, an MNN-RG with a poorly generalized NN  may still fail to track the desired reference trajectory.

\subsection{Robustness to model uncertainty, noise, and disturbance}
\label{Robustness}
Model uncertainties, noise, and disturbances can deteriorate the performance of model-based controllers. Here, to evaluate the robustness of the proposed controller under practical conditions, the MNN-RG based on a reduced third-order model from the previous section, is applied to a higher-fidelity nine-state full-order representation of the system \cite{pukrushpan2002modeling}. We assume, the state feedback to MNN-RG is corrupted by sensor noise, where the pressure and compressor motor speed noise are modeled as white Gaussian noise with standard deviations of $500$ Pa and $5$ rad/s, respectively. Also, to model uncertainties and disturbances associated with temperature and humidity fluctuations in a real system, the atmospheric relative humidity is modified from $0.5$ to $0.25$, and the temperature of the stack is modified from $80 ^\circ\text{C}$ to $85 ^\circ\text{C}$. 

The simulation results under the above conditions are presented in Fig.~\ref{fig:9th-order}. As observed, under the given conditions, the MNN-RG exhibits a fairly robust performance in tracking the current demand. However, compared to its performance on the nominal model (see Fig.~\ref{fig:9}(a)), the OER violates the constraint by $0.1$ ($5\%$) below the required threshold of $1.9$ during transients, and the compressor trajectory near the surge boundary shows a small excursion into the surge region, indicating slight sensitivity to model mismatch, noise, and disturbances. We emphasize that  applying the PRG to the FC system under the given conditions produce nearly identical results, which are omitted due to space limitations. Thus, the MNN-RG has essentially the same robustness properties as the underlying PRG.

While the performance of the MNN-RG is fairly robust under the given conditions, we acknowledge that constraint violations become more significant as the uncertainties increase. Since the MNN-RG can inherit the robustness properties of the underlying PRG, achieving robustness in constraint satisfaction under such circumstances requires designing a ``robust PRG''. This can be accomplished by incorporating model and measurement uncertainties into the system model (1)-(2) for PRG design  followed by a min-max optimization to handle worst-case scenarios or a scenario-based optimization to improve robustness across multiple uncertainty realizations. However, robust PRG has already been studied in the literature \cite{sun2005load} and is outside the scope of the current paper. 
\begin{figure}
    \centering
    \includegraphics[width=0.94\linewidth]{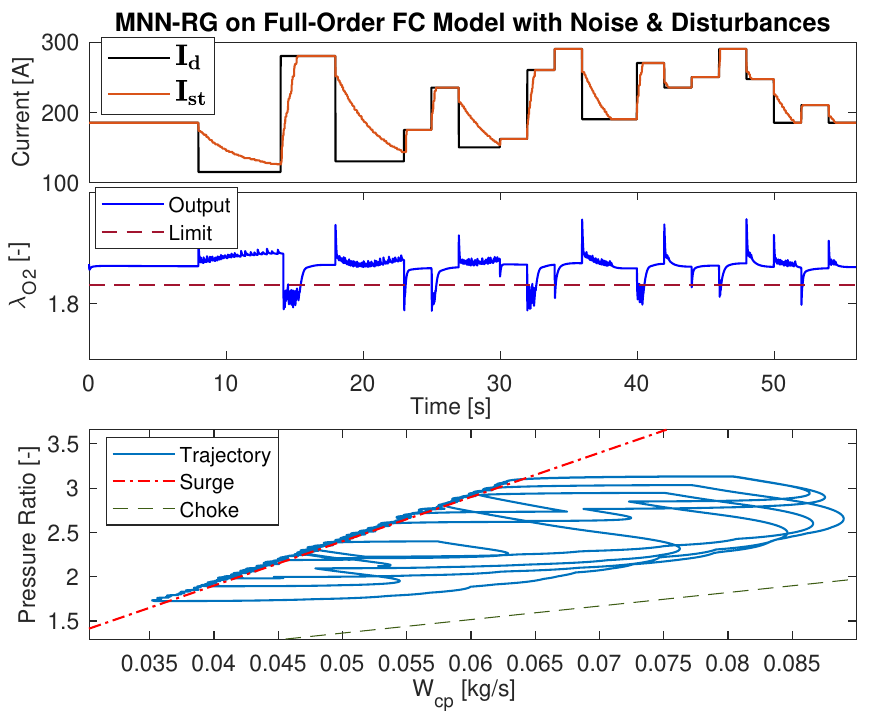}
    \caption{Robustness to model mismatch, noise, and disturbances.}
    \label{fig:9th-order}
\end{figure}  
\subsection{Drive-cycle simulation and computational time analysis}
\label{Drive-cycle simulation}
In this section, we evaluate the performance of the MNN-RG with $\bar{M}$-tuning, trained on the partial data-set, on a realistic drive-cycle. The setpoint used for this validation is representative of the FC system power needed for the completion of a standard drive-cycle, referred to as “Dynamic Test”, as shown Fig.~\ref{fig:12}.
\begin{figure}
\centering 
\begin{subfigure}{0.45\textwidth}
            \includegraphics[width=\textwidth]{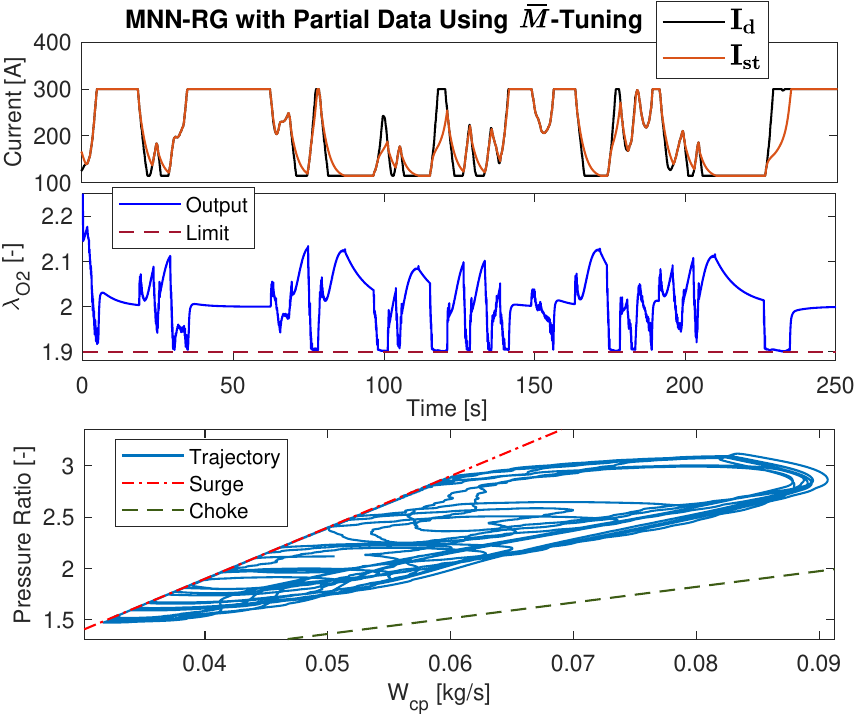}
    \end{subfigure}
\caption{ Drive-cycle simulation with MNN-RG trained on the partial dataset.}
\label{fig:12}
\end{figure}
As observed in Fig.~\ref{fig:12}, the proposed MNN-RG successfully maintains the OER above $1.9$ and avoids compressor surge and choke while pushing the system to its limits. We also simulated the FC system on the same drive-cycle using NN-RG and PRG. For space limitations, detailed simulation plots for the NN-RG and PRG controllers are not shown. However, the key findings are:
\begin{itemize}
    \item The NN-RG, even the well-generalized one trained on the comprehensive dataset, fails to produce a load response that satisfies the constraints at all times, similar to the results presented in Section~\ref{NN-RG for load control}. 
    \item MNN-RG effectively satisfies all the constraints and aligns well with the load response of the PRG, achieving an RMSE of $0.1748$ [A].
\end{itemize}

Additionally, we present a computational comparison between the  PRG, the well-generalized NN-RG trained on the comprehensive dataset, and the proposed MNN-RG trained on the partial dataset. All simulations are conducted in MATLAB on a laptop equipped with a 12\textsuperscript{th} Gen Intel(R) Core(TM) i9-12900H @ 2.90 GHz processor and 16 GB of memory. MATLAB software utilizes a stiff variable-step, variable-order solver \ttfamily{ode15s} \normalfont with a tolerance of $10^{-5}$. Since the MNN-RG optimization problem \eqref{eqn:14} consists of 1000 quadratic inequalities, we opted to solve the optimization problem using a bisectional search with $L=15$, as discussed in Remark~\ref{rm:twowaysolution}. 

For the purpose of timing the computations, the simulation for each method is executed $10$ times to minimize variability from the computer's background processes. 
For each timestep, the minimum computational time across these $10$ simulations is selected. The average execution time, presented in Table~\ref{tab:Performance of different load governors}, is then calculated as the sum of these minimum times divided by the total number of timesteps in the simulation. Additionally, the worst-case execution time, defined as the maximum  among these minimum times across all timestep, is also listed in Table~\ref{tab:Performance of different load governors}. For a controller to be executed in real time, the worst-case execution time must be less than the sampling interval. 
\begin{table}
\centering
\begin{threeparttable}
\caption{\\ \uppercase{Performance of Different Load Governors}}
\footnotesize 
\begin{tabular}{|c|c|c|c|}
\hline
\multirow{2}{*}{\textbf{Case}} & \multicolumn{1}{c|}{Average} & \multicolumn{1}{c|}{Worst-case} & \multicolumn{1}{c|}{Constraints} \\ 
& \multicolumn{1}{c|}{execution time [ms]} & \multicolumn{1}{c|}{ execution time [ms]} & \multicolumn{1}{c|}{satisfied?} \\ \hline
NN-RG & $1.82$ & $3.1$ & No \\ \hline
PRG & $7.68$ & $22.9$ & Yes \\ \hline
MNN-RG & $4.19$ & $5.8$ & Yes \\ \hline
\end{tabular}
\label{tab:Performance of different load governors}
\begin{tablenotes}
\item \hspace{-1em} \textbf{Note}: The NN-RG and MNN-RG represent the models trained on the comprehensive and partial datasets, respectively.
\end{tablenotes}
\end{threeparttable}
\end{table}
The execution time for MNN-RG includes the time taken for the NN to generate its outputs, the time required to solve the nominal state equations and state sensitivity equations given in \mbox{Appendix B}, as well as the time to solve \eqref{eqn:14} using  the bisectional search. 
According to the results in Table IV, the following points about each method can be concluded:
\begin{itemize}
    \item The PRG method, while satisfying all the constraints, exhibits the highest average and worst case execution times, exceeding the sampling time and indicating real-time execution issues. 
    \item The NN-RG approach has the lowest average and worst-case execution times (the worst case being within the sampling time), but it fails to meet the constraints.
    \item The MNN-RG framework strikes a balance with a moderate average execution time of $4.19$ ms and a worst-case execution time of $5.8$ ms, successfully satisfying the constraints.
\end{itemize}

In MNN-RG, the execution time consistently remains lower than the sampling time. However, in practice, additional computational simplifications may be required, as embedded hardware typically has significantly less processing power compared to a desktop computer. In resource-constrained environments, if MNN-RG is not real-time implementable,  its complexity can be reduced by computing a smaller admissibility index $j^{\ast}$ using the level set of a Lyapunov function, as mentioned in Section~\ref{Problem setup}, and/or it can be run at a slower sample rate, although at the expense of reduced control performance.

The higher computational demand of PRG compared to the proposed MNN-RG stems from the need to simulate the nonlinear dynamics multiple times at each timestep until the bisectional search concludes. In contrast, MNN-RG at each timestep solves the nominal state equation and state sensitivity equations of the air-path system just once, followed by the bisectional search over the quadratic inequalities to determine the optimal gain.
\begin{remark}
While the proposed MNN-RG exhibits a worse-case execution time roughly $4$ times smaller than the PRG for this FC system, its computational advantages of PRG are expected to be even more pronounced for higher-dimensional, more complex systems and for systems requiring larger prediction horizons, such as those with relatively slow dynamics. This is because the simulations in the PRG algorithm (see Step 6 in Algorithm 1) and, thus, the execution time of the PRG as a whole, grows with both the prediction horizon and the system's state dimension.  
\end{remark}

Note that in this case study, both the MNN-RG and PRG employ the bisectional search for the online evaluation of $\kappa$. While increasing the number of iterations, $L$, in the bisectional search yields higher precision in evaluation of $\kappa$, it also leads to an increase in the worst-case execution time of both RGs, as shown in Fig.~\ref{fig:13}. 
\begin{figure}
\centering
\includegraphics[width=0.4\textwidth]{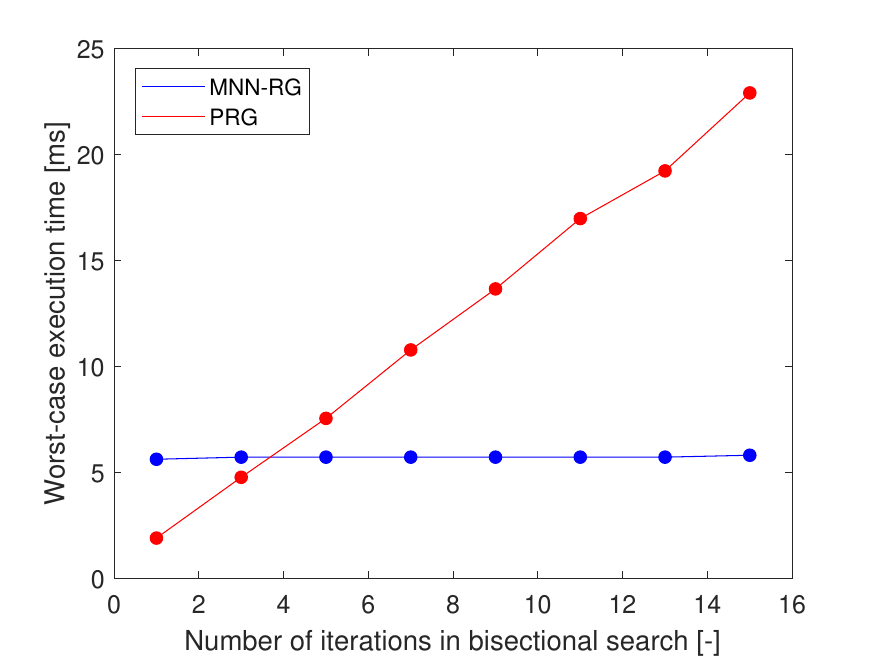}
\caption{Worst-case execution time comparison between PRG and MNN-RG as a function of bisectional search iterations, $L$.}
\label{fig:13}
\end{figure}
The plot reveals two interesting points: Firstly, as $L$ increases, the execution time for the PRG rises linearly (as discussed in Remark 6). In contrast, while MNN-RG execution time also increases, the increase is negligible. This suggests that the computational cost of evaluating the quadratic inequalities in \eqref{eqn:14} is insignificant compared to the operations that are independent of $L$, namely solving the sensitivity and nominal state equations. This makes MNN-RG a preferable choice for real-time control systems where maintaining a predictable and low execution time is essential. Secondly, for values of $L$ below 4, where the precision in $\kappa$ is low, PRG has a shorter execution time, while for larger values of $L$, PRG has a higher execution time. 
This highlights that the MNN-RG is better-suited for higher precision requirements.  
\begin{remark}
It is worth noting that other function approximators, such as Gaussian Process Regression, can be used in place of NN within the MNN-RG framework. The key requirement is to ensure that the chosen model demonstrates acceptable generalization capabilities, as discussed in this section.
\end{remark}
\section{Conclusions and outlook}
In this paper, a computationally-efficient machine learning-based constraint management method for nonlinear systems was developed based on the reference governor (RG) framework. The main objective was to enforce output constraints with minimal computational burden. While the  prediction-based nonlinear RG (PRG) from the literature can enforce constraints on the nonlinear system, it relies on multiple online simulations of the nonlinear dynamics at each sample time, making it computationally demanding. The proposed approach, referred to as the Modified Neural Network RG (MNN-RG), employs a neural network (NN) to approximate the input-output map of the PRG, followed by a modification to guarantee constraint satisfaction despite NN training errors. This modification is achieved using a novel sensitivity-based strategy. 

The proposed approach was then used as a load governor to enforce constraints on oxygen starvation, compressor surge, and compressor choke of an automotive fuel cell system. Through simulation studies, we demonstrated the impacts of NN generalization on the tracking and constraint enforcement capabilities of the MNN-RG. We observed that the MNN-RG improves online execution time compared to the PRG; however, the extent of this improvement is case-dependent.

Future work will focus on improving the robustness of MNN-RG in constraint satisfaction under model mismatch, noise, and disturbances. 
Additionally, experimental validation will be performed on a more complex constrained system with high-dimensional dynamics to evaluate both the real-world applicability and scalability of MNN-RG. The use of advanced machine learning models within MNN-RG will also be explored to enhance its effectiveness. For the fuel cell system, future studies will address additional constraints, including hydration and temperature, to ensure more reliable operation.
\label{conclusions}
\section*{Appendix A \\ proof of Theorem 1}
\label{thm:MNNRG propeties}
To prove constraint enforcement, consider the following: At each timestep, the MNN-RG uses the most recent $x(t)$ and $r(t)$ to solve \eqref{eqn:14}, which results in a \( v(t) \) that satisfies  \( \bar{y}_v \leq -\varepsilon \) and the condition \eqref{eqn:13} over the prediction horizon. Since \eqref{eqn:13} is a tightened version of \( \hat{y}(j) \leq 0 \), satisfaction of \eqref{eqn:13} by \( v(t) \) implies \( \hat{y}(j) \leq 0 \) over the prediction horizon. Thus, \( (x(t), v(t)) \in \Omega \), which, by the finite constraint horizon property of $\Omega$, implies that  \( \hat{y}(j) \leq 0 \) is satisfied for all future times.

To prove the convergence property, assume $r(t)$ remains constant for $t\geq t_0$, i.e., $r(t)=r$ $\forall t\geq t_0$. From the update law of MNN-RG (which is the same as Eq. \eqref{eqn:5}), and with \( \kappa \in [0, 1] \), it follows that \( v(t) \) is a convex combination of \( r(t) \) and \( v(t - 1) \), both of which are bounded. Therefore, $v(t)$ forms a monotonic sequence bounded by $r$, which implies convergence.

ISS stability, as assumed in the theorem, means that the system's state remains bounded for bounded inputs and converges to zero as the input approaches zero. Mathematically, 
 there exist a class \(\mathcal{KL}\) function \(\beta\) and a class \(\mathcal{K}\) function \(\gamma\) such that for any initial state \( x(0)\in\mathbb{R}^{n_x} \) and any bounded input \( v(t) \), the solution \( x(t) \)  exists for all $t \geq 0$ and satisfies $\|x(t)\| \leq \beta(\|x(0)\|, t) + \gamma(|v|_\infty)$
 \begin{flalign}
\|x(t)\| \leq \beta(\|x(0)\|, t) + \gamma(|v|_\infty),
\label{eqn: 0112}
\end{flalign}
where  \(|v|_\infty = \sup \{ |v(t)| : t \in \mathbb{Z}_+ \} < \infty\), see \cite{jiang2001input} for details.
To prove ISS stability in the context of MNN-RG framework, we first assume that both \( v(-1) \)  and \( r(t) \) are bounded. 
Using the update law given in \eqref{eqn:5}, \( v(t) \) remains bounded for all \( t \geq 0 \) (\(|v|_\infty < \infty\)). 
To demonstrate that the ISS property is preserved with respect to \( r(t) \), we note that \(|v|_\infty = \alpha |r|_\infty \) for some \(\alpha \geq 0\). 
Substituting this into \eqref{eqn: 0112} gives \(\|x(t)\| \leq \beta(\|x(0)\|, t) + \gamma(\alpha |r|_\infty) \). Defining \( \gamma'(\cdot): \mathbb{R} \rightarrow \mathbb{R}\) as \(\gamma'(s) = \gamma(\alpha s) \), we have \(\|x(t)\| \leq \beta(\|x(0)\|, t) + \gamma'(|r|_\infty) \). Since \(\gamma' \) is a composition of a \(\mathcal{K}\)-class function \(\gamma \) and a positive scalar multiplication, it remains a \(\mathcal{K}\)-class function, thereby preserving the ISS property with respect to \( r(t) \), which implies ISS stability of the entire system.

\section*{Appendix B}
This appendix provides the nominal state and state sensitivity equations of the third-order model for the air-path of a FC system developed in \cite{talj2009experimental}:
\begin{flalign*}
\begin{aligned}
\dot{\hat{x}}_{1n} &= -\mu_1 \hat{x}_{1n} + \mu_2 \hat{x}_{3n} + \mu_3 - \mu_4 v_n\\
\dot{\hat{x}}_{2n} &= c_{10} \left( g_1 v_n + g_2 + k_p \left( c_{15} v_n - W_{cp,n} \right) + k_i \hat{x}_{4n} \right)\\
& \quad - c_6 \hat{x}_{2n} - \frac{c_{7}}{\hat{x}_{2n}} \left[ \left( \frac{\hat{x}_{3n}}{c_{8}} \right)^{c_{9}} - 1 \right] W_{cp,n}\\
\dot{\hat{x}}_{3n} &= c_{11} \left( 1 + c_{12} \left[ \left( \frac{\hat{x}_{3n}}{c_{8}} \right)^{c_{9}} - 1 \right] \right)\\
& \quad \times \left( W_{cp,n} - c_{14} \left( \hat{x}_{3n} - \hat{x}_{1n} \right) \right)\\
\dot{\hat{x}}_{4n} &= c_{15} v_n - W_{cp,n}\\
\dot{S}_{x1} &= -\mu_1 S_{x1} + \mu_2 S_{x3} - \mu_4\\
\end{aligned}
\end{flalign*}
\begin{flalign*}
\begin{aligned}
\dot{S}_{x2} &= S_{x2} \biggl\{ -c_{10} k_p \frac{\partial W_{cp}}{\partial x_{2}} - c_6 - \frac{c_{7}}{x_{2}} \left[ \left( \frac{x_{3}}{c_{8}} \right)^{c_{9}} - 1 \right] \\& \quad \times \frac{\partial W_{cp}}{\partial x_{2}}+ \frac{c_{7}}{x_{2}^2} \left[ \left( \frac{x_{3}}{c_{8}} \right)^{c_{12}} - 1 \right] W_{cp} \biggr\}_{x=\hat{x}_{n}}\\
& \quad + \biggl\{ -c_{10} k_p \frac{\partial W_{cp}}{\partial x_{3}} - \frac{c_{7}}{x_{2}} \left( \frac{c_{9} x_{3}^{c_{9}-1}}{c_{8}^{c_{9}}} \right) W_{cp}\\
& \quad - \frac{c_{7}}{x_{2}} \left[ \left( \frac{x_{3}}{c_{8}} \right)^{c_{9}} - 1 \right] \frac{\partial W_{cp}}{\partial x_{3}} \biggr\}_{x=\hat{x}_{n}} + c_{10} k_i S_{x4} \\& 
\quad + c_{10} (g_1 + k_p c_{15})\\
\dot{S}_{x3} &= S_{x1} c_{14} c_{11} \left\{ 1 + c_{12} \left[ \left( \frac{x_{3}}{c_{8}} \right)^{c_{9}} - 1 \right] \right\}_{x=\hat{x}_{n}}\\
& \quad + S_{x2} \left( c_{11} \left\{ 1 + c_{12} \left[ \left( \frac{x_{3}}{c_{8}} \right)^{c_{9}} - 1 \right] \right\} \frac{\partial W_{cp}}{\partial x_{2}} \right)_{x=\hat{x}_{n}}\\
& \quad + S_{x3} \left( \frac{c_{11} c_{9} c_{12} x_{3}^{c_{9}-1}}{c_{8}^{c_{9}}} \left( W_{cp} - c_{14} \left( x_{3} - x_{1} \right) \right)\right.\\
& \quad + c_{11} \left\{ 1 + c_{12} \left[ \left( \frac{x_{3}}{c_{8}} \right)^{c_{9}} - 1 \right] \right\} \left. \left( \frac{\partial W_{cp}}{\partial x_{3}} - c_{14} \right)\right)_{x=\hat{x}_{n}}\\
\dot{S}_{x4} &= -\left. \frac{\partial W_{cp}}{\partial x_{2}} \right|_{x=\hat{x}_{n}} S_{x2} - \left. \frac{\partial W_{cp}}{\partial x_{3}} \right|_{x=\hat{x}_{n}} S_{x3} + c_{15}
\end{aligned}
\end{flalign*}
The PI gains are $k_p=100$, $k_i=500$, with the static feedforward control parameters being $g_1=0.6814$ and $g_2=33.8741$. The expressions of parameters $\mu_i$ ($1\le i\le 4$) and $c_i$ ($1\le i\le 17$) are shown in Appendix D. The compressor flow, $W_{cp}$ is obtained by the Jensen-Kristensen model, which relies on nonlinear curve fitting. As such, the partial derivatives of $W_{cp}$ in the above expressions must be evaluated numerically.
\section*{Appendix C}
This appendix provides the sensitivities of the constrained outputs with respect to the reference in the air-path system.
\begin{flalign*}
    \begin{split}
        S_{y_{1}}(j)  &=-\frac{c_{17}}{c_{16} v_n} S_{x1}(j) + \frac{c_{17}}{c_{16} v_{n}} S_{x3}(j) - \frac{c_{17} (\hat{x}_{3n}(j) - \hat{x}_{1n}(j))}{c_{16} v_{n}^{2}} \\ 
        S_{y_{2}}(j) &= \left. \frac{\partial W_{cp}}{\partial x_{2}}(j) \right|_{x=\hat{x}_{n}} S_{x2}(j) + \left. \frac{\partial W_{cp}}{\partial x_{3}}(j) \right|_{x=\hat{x}_{n}} S_{x3}(j) \\
        S_{y_{3}}(j) &= S_{x3}(j)
    \end{split}
\label{eqn:29}
\end{flalign*}
\section*{Appendix D}
In this Appendix, we define the constants \( \mu_i \), \( i = 1, \ldots, 4 \) and \( c_i \), \( i = 1, \ldots, 17 \), that appear in the FC state equation and state sensitivity equations.
\begin{align*}  
\chi &= 0.0264,
\mu_{1} = 0.88\left(c_{1} + c_{5} + c_{3}\cdot c_{13}/\chi\right), 
\mu_{2} = c_{1} + c_{5}, \\
\mu_{3} &= 0.88\left(c_{2} c_{3} c_{13}/\chi\right), \mu_{4} = c_{4}, c_{1} = \frac{R_u T_{st} k_{ca,in} x_{O_2,atm}}{M_{O_2} V_{ca}(1 + \omega_{atm})}, \\
c_{2} &= p_{sat}(T_{st}), c_{3} = \frac{R_u T_{st}}{V_{ca}}, c_{4} = \frac{n_{cell} R_u T_{st}}{4 V_{ca} F}, \\
c_{5} &= \frac{R_u T_{st} k_{ca,in}(1 - x_{O_2,atm})}{M_{N_2} V_{ca}(1 + \omega_{atm})}, c_{6} = \frac{\eta_{cm} k_t k_v}{J_{cp} R_{cm}},\\
c_{7} &= \frac{C_p T_{atm}}{J_{cp} \eta_{cp}}, 
c_{8} = p_{atm}, c_{9} = \frac{\gamma_{atm} - 1}{\gamma_{atm}}, c_{10} = \frac{\eta_{cm} k_t}{J_{cp} R_{cm}} ,\\
\end{align*}
\begin{align*}
c_{11} &= \frac{R_u T_{atm}}{M_{a,atm} V_{sm}},
c_{12} = \frac{1}{\eta_{cp}},\\
c_{13} &= \frac{C_{D} A_t}{\sqrt{R_u T_{st}}}{\gamma_{atm}}^{\frac{1}{2}} \left( \frac{2}{\gamma_{atm} + 1} \right)^{\frac{\gamma_{atm} + 1}{2(\gamma_{atm} - 1)}}, \\
c_{14} &= k_{ca,in}, c_{15} = \frac{c_{16} \lambda _{O_{2} ,des}(1 + \omega_{atm})}{x _{O_{2} ,atm}} , c_{16} = \frac{n_{cell} M_{O_2}}{4 F},\\
c_{17} &= k_{ca,in} \frac{x_{O_2,atm}}{1 + \omega_{atm}}
\end{align*}
Please see \cite{pukrushpan2004control} for details and the numerical values of the parameters.

\section*{Acknowledgment}
The authors would like to thank Zeng Qiu, Hao Wang, Chris Weinkauf, and Michiel Van Nieuwstadt from Ford Motor Company, Dearborn, MI, USA, for their
technical comments during the course of this study. We acknowledge funding from Ford Motor Company under a URP Award 002092-URP.

\bibliographystyle{IEEEtran}
\bibliography{conference_101719}

\begin{IEEEbiography}[{\includegraphics[width=1in,height=1.25in,clip,keepaspectratio]{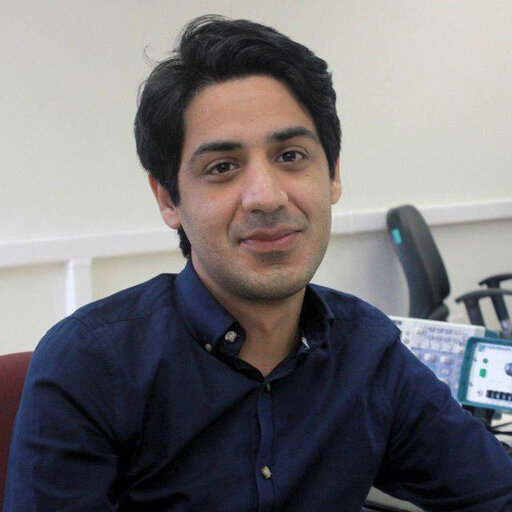}}]{Mostafaali Ayubirad}
(Member, IEEE)
received the M.Sc. degree in electrical engineering (control systems) from the University of Tehran, Tehran, Iran, in 2015. He is currently pursuing the Ph.D. degree in electrical and biomedical engineering with the University of Vermont, Burlington, VT, USA.

His research interests lie in systems and control theory, with a focus on predictive and constrained control, particularly in applications related to automotive fuel cell vehicles.

Mr. Ayubirad's academic excellence has been recognized with the Cyril G. Veinott (Graduate) Award for outstanding academic performance and potential in academia.
\end{IEEEbiography}
\vspace{-13.5cm}
\begin{IEEEbiography}[{\includegraphics[width=1in,height=1.25in,clip,keepaspectratio]{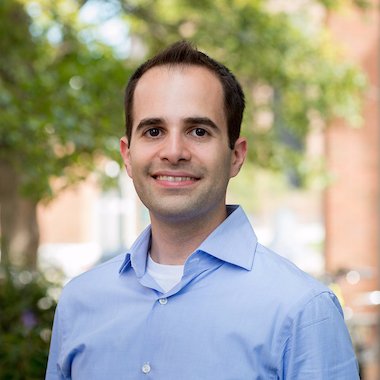}}]{Hamid R. Ossareh}
(Senior Member, IEEE) received the B.A.Sc. degree from the University of Toronto, Toronto, ON, Canada, in 2008, and the Ph.D. degree from the University of Michigan, Ann Arbor, MI, USA, in 2013.

From 2013 to 2016, he was with Ford Research and Advanced Engineering, Dearborn, MI, USA, as a Research Engineer. Since 2016, he has been a Faculty Member with the University of Vermont (UVM), Burlington, VT, USA. He is currently an Associate Professor. He holds several patents. His research interests include the areas of systems and control theory, more specifically predictive control, nonlinear control, and constrained control, with application areas of automotive, power, aerospace, and xerographic systems.

Dr. Ossareh has won several awards, including the Faculty of the Year Award from IEEE Green Mountain Section, the Excellence in Research Award from UVM College of Engineering and Mathematical Sciences, and the Ford Technical Achievement Award from Ford Motor Company.
\end{IEEEbiography}

\end{document}